\documentclass[12pt, draftclsnofoot, onecolumn]{IEEEtran}
\usepackage{epsfig,latexsym}
\usepackage{float}
\usepackage{indentfirst}
\usepackage{amsmath}
\usepackage{amssymb}
\usepackage{times}
\usepackage{subfigure}
\usepackage{psfrag}
\usepackage{hyperref}
\usepackage{cite}
\usepackage{lastpage}
\usepackage{fancyhdr}
\usepackage{color}
 \usepackage{amsthm}
\usepackage{bigints}
\newtheorem{theorem}{Theorem}
\newtheorem{Lemma}{Lemma}
\newtheorem{lemma}[Lemma]{$\mathbf{Lemma}$}

\begin{document}
\title{ NOMA Assisted Wireless Caching: \\Strategies and Performance Analysis}

\author{ Zhiguo Ding, \IEEEmembership{Senior Member, IEEE}, Pingzhi Fan,   \IEEEmembership{Fellow, IEEE}, George K. Karagiannidis, \IEEEmembership{Fellow, IEEE}, Robert Schober, \IEEEmembership{Fellow, IEEE},  \\ and H. Vincent Poor, \IEEEmembership{Fellow, IEEE} \\({\it Invited Paper})\thanks{ 
This work was  presented  in part  at the {\it IEEE International Conference on Communications (ICC)}, Kansas City, MO, May 2018 \cite{Dingicc18}. 

    Z. Ding and H. V. Poor are  with the Department of
Electrical Engineering, Princeton University, Princeton, NJ 08544,
USA. Z. Ding
 is also  with the School of
Computing and Communications, Lancaster
University, Lancaster, UK (email: \href{mailto:z.ding@lancaster.ac.uk}{z.ding@lancaster.ac.uk}, \href{mailto:poor@princeton.edu}{poor@princeton.edu}).
 
P. Fan is with the Institute of Mobile
Communications, Southwest Jiaotong University, Chengdu, China (email: \href{mailto:pingzhifan@foxmail.com}{pingzhifan@foxmail.com}).

G. K. Karagiannidis is with the Department of Electrical and Computer
Engineering, Aristotle University of Thessaloniki, Thessaloniki, Greece (email: \href{mailto:geokarag@auth.gr}{geokarag@auth.gr}).

R. Schober is with the Institute for Digital Communications,
Friedrich-Alexander-University Erlangen-Nurnberg (FAU), Germany (email: \href{mailto:robert.schober@fau.de}{robert.schober@fau.de}).

  }\vspace{-3em}} \maketitle

\begin{abstract}
Conventional wireless caching assumes that content can be pushed to local caching infrastructure during off-peak hours in an error-free manner; however, this assumption  is not applicable if   local caches    need to be frequently updated via wireless transmission.  This paper investigates  a new approach to wireless caching   for the case when cache content has to be updated during on-peak hours.  Two non-orthogonal multiple access (NOMA) assisted  caching strategies are developed, namely the {\it push-then-deliver strategy} and the {\it push-and-deliver strategy}.   In the push-then-deliver strategy,    the NOMA principle is applied to push more content files to the content servers during a short time interval   reserved for content pushing in on-peak hours and to provide more connectivity for content delivery, compared to     the conventional orthogonal multiple access (OMA)    strategy. The push-and-deliver strategy is motivated by the fact that some users' requests cannot be accommodated locally and   the base station has to serve them directly. These events during the content delivery phase are exploited  as opportunities for content pushing, which further facilitates the frequent update of the files cached at the content servers.  It is also shown   that this  strategy can be straightforwardly extended to device-to-device   caching, and various   analytical results are developed to illustrate  the superiority  of the proposed  caching strategies compared to OMA based schemes.  
\end{abstract} 
\section{Introduction}
Recently non-orthogonal multiple access (NOMA) has received significant attention as a key enabling technique for future wireless networks  \cite{jsacnomaxmine, docom, teckuk}. The key idea of NOMA is to encourage spectrum sharing among mobile nodes, which not only improves the spectral efficiency but also ensures that massive connectivity can be effectively supported.   Practical   concepts for implementing the NOMA principle  for a single resource block, such as an orthogonal frequency division multiplexing (OFDM)  subcarrier,  include  power domain NOMA and cognitive radio (CR) inspired NOMA \cite{6692652,Nomading, Zhiguo_CRconoma}, which   provide  different tradeoffs between throughput and fairness. When each user is allowed to occupy multiple   subcarriers, dynamically  grouping the users on different subcarriers  is a challenging problem, and  various  multi-carrier  NOMA schemes, such as sparse code multiple access (SCMA) and pattern division multiple access (PDMA)  \cite{6666156, 7024798}, provide practical solutions for  achieving different   performance-complexity tradeoffs. Unlike single-carrier NOMA, in multi-carrier NOMA, a user's message is spread over multiple resource blocks, which requires efficient encoding schemes, such as multi-dimensional coding, to be implemented at the transmitter and low-complexity decoding schemes, such as   message passing algorithms, to be used at the receivers. 

NOMA has   been shown to be compatible  with many   other advanced  communication concepts. For example, several  features of millimeter-wave (mmWave) communications, such as   highly directional    transmission,  and  the mismatch between the users' channel vectors and the commonly  used finite resolution analog beamforming, facilitate the implementation of NOMA in mmWave networks \cite{Zhiguo_mmwave,7918554}. In addition, NOMA can  further improve the spectral efficiency of multiple-input multiple-output (MIMO) systems. For example,   MIMO-NOMA can efficiently  exploit  the spatial degrees of freedom of MIMO channels and, unlike single-input single-output (SISO) NOMA, is beneficial  even if the users have similar channel conditions \cite{7015589, Zhiguo_mimoconoma, jsacnoma2}. Furthermore, conventionally, when the users have a single antenna, cooperative transmission can be used to exploit spatial diversity but    suffers from a reduced  overall data rate, since relaying    consumes extra bandwidth resources \cite{Laneman04}. In this context, the application of NOMA   can efficiently reduce the number of   consumed bandwidth resource blocks, such as subcarriers and time slots, and hence improve the spectral efficiency of cooperative communications  \cite{Zhiguo_conoma, jsacnoma23, 7482799}.   Furthermore,   existing studies have also revealed a strong synergy between NOMA and CR networks, where the use of NOMA can significantly  improve the connectivity for the users of the secondary network  \cite{7398134}. 

Wireless caching    is another important enabling technique for future communication  networks \cite{6871674,6495773}, but little is known about the coexistence of NOMA and wireless caching. The key idea of wireless caching is to push the content in off-peak hours during the so-called  {\it content pushing phase} close to the users before it is requested,  and therefore, the users' requests can   be locally served during the so-called {\it content delivery phase}. In fact,   asking a base station (BS)   to serve the users' requests directly  is not preferable, not only because the maximal number of users that a BS can serve concurrently is small, but also because   non-caching transmission schemes  are severely constrained by the limited backhaul capacity   of wireless networks.     Most caching schemes  can be grouped into one of two classes  \footnote{We note that  coded caching, where the number of BS transmissions  is reduced by exploiting the structure of the content sent during the content pushing and delivery phases, does not fall into the two considered categories \cite{6763007,7932139}. } \cite{6495773, 7536874}. {\it The first class}  assumes the existence of a   content caching infrastructure, such as content servers, small cell BSs, etc., \cite{7488289,7565184,7828114}.  When   caching infrastructure (e.g., content servers) is available, the objective  in the content pushing phase is to  push the content files to the content servers in a timely and reliable  manner, before the users request these files. During the phase of content delivery, an ideal situation is that all the users' requests can be locally served, without communicating with the central controller of the network, e.g., the BS. {\it The second class},   also known as device-to-device (D2D) caching,   assumes that there is no dedicated caching infrastructure, and relies on user cooperation \cite{7342961,7932468}. Particularly, during the content pushing phase, all   users will proactively cache some content. During the content delivery phase,  a user will communicate with its BS only if none of its neighbours can help the user locally, i.e., the user cannot find its requested file in the caches of its surrounding neighbours. 

 A fundamental  assumption made in the existing caching literature   is that, in the content pushing phase, content is  pushed to the content servers in an \underline{error-free} manner during  \underline{off-peak} hours.  However, performing caching only during    off-peak hours is not effective  if the popularity of the content is rapidly changing or the files to be cached need to be frequently updated.  Typical examples for this type of content include up-to-the-minute news,    sports events requiring  live updates,   e-commerce promotion  with frequent pricing changes, newly released music videos,  etc. Similarly, assuming   error-free content pushing   may also be questionable   in many practical communication scenarios.    In practice, connecting the content servers wirelessly with the BS is preferable since   the cost for setting up the network is reduced  and  the installation of cables is avoided. Furthermore, wireless networks facilitate  D2D caching, since file sharing among users in wireless networks is straightforward, whereas realizing pairwise connections in wireline   networks is more difficult.  However, wireless transmission is prone to noise,   distortion, and attenuation, which makes error-free transmission a very strong assumption in practice. The objective  of this paper is to apply the NOMA principle to wireless caching and to develop NOMA assisted caching strategies which do not require  the aforementioned   assumptions.  The contributions of the paper are summarized as follows:

\begin{itemize}
\item  For the case where the content pushing and delivery phases are separated and limited bandwidth resources are periodically available for content pushing in on-peak hours, a NOMA-assisted push-then-deliver strategy is proposed. Particularly, during the content pushing phase, the BS will use the NOMA principle and push multiple files to  the  content servers simultaneously.  A CR inspired NOMA power allocation policy is used to ensure that content files are delivered to  their target content servers with the same outage probability as   with conventional orthogonal multiple access (OMA) based transmission. However, by using NOMA,   additional files can be  pushed to the content servers simultaneously,  which   is important to efficiently use the limited resources   reserved for content pushing and hence to improve the cache hit probability. During the content delivery phase, the use of NOMA not only improves the reliability of content delivery, but also ensures that more user requests can be served concurrently by a content server. 

\item  The objective of the proposed push-and-deliver strategy is     to provide additional bandwidth resources   for content pushing.  Unlike the push-then-deliver strategy, the push-and-deliver strategy   seeks opportunities for content pushing during the content delivery phase. In particular,   during the content delivery phase, the BS  occasionally has to  serve some  users directly, since these users' requested files cannot be found in the local content severs. Conventionally, this is a non-ideal situation which reduces the spectral efficiency. Nevertheless,  this non-ideal situation is inevitable in practice and is expected
to occur frequently, as the users' requests cannot be perfectly predicted. In this paper, this non-ideal situation for  content delivery   is exploited  as an opportunity  for additional content pushing. In other words, the push-and-deliver strategy is particularly useful  when  the bandwidth resources   reserved for content pushing are limited, but the files at the content servers need to be frequently  updated. The proposed push-and-deliver strategy is  also   extended to   D2D caching   without caching infrastructure, where the caches at the D2D helpers can be refreshed while users are directly served by the base station.   We note  that the NOMA-multicasting scheme proposed in \cite{jsacnoma22} can be viewed as a D2D special case of the  proposed push-and-deliver strategy, if the multicasting phase   in  \cite{jsacnoma22} is viewed as the content delivery phase. However,  the impact of integrating content pushing and delivery on the cache hit probability was not   investigated  in \cite{jsacnoma22}. 

\item Analytical results for the cache hit probability,  the transmission outage probability, and the D2D cache miss probability are derived  in order to obtain a better  understanding of the performance of the proposed caching strategies.  Conventionally, the cache hit probability is mainly determined by the size of the caches of the content servers, instead of by   transmission outages, since conventional  content pushing is carried out during  off-peak hours, which means that the amount of the pushed content is much larger than the size of the caches of the content servers. However, for the schemes proposed in this paper,   the outage based cache hit probability is a more suitable metric  for performance evaluation, as explained in the following. In particular,  we assume that   a short time interval  is periodically reserved  during on-peak hours for   pushing new content   to the content servers. The  time interval  reserved for content pushing has to be short in order to achieve high  spectral efficiency  and has to be   shared by multiple  content servers. Thus, the amount of content   that can be  pushed to the content servers may be much  smaller than the size of the storage of the content servers. Considering   the tremendous increase in storage capacity available with current technologies, this is a realistic assumption. For the considered caching scenario, the crucial issue is how to quickly push the content files to the content servers during the short time interval available for content pushing during on-peak hours. Therefore,  the outage based cache hit probability is the relevant   performance criterion. When   caching infrastructure is available, the impact of NOMA on the content pushing phase is quantified by exploiting  the   joint probability density function (pdf) of the distances   between  the content servers and the BS, and    closed-form expressions for  the achieved cache hit probability are developed. The impact of NOMA on the content delivery phase is investigated by using the transmission outage probability as a performance  criterion and modelling the locations of the users and the content servers as    Poisson cluster processes (PCPs). Furthermore, the impact of NOMA on D2D caching is studied by modelling the effect of content pushing as a thinning Poisson point process and deriving  the cache miss probability, i.e., the probability  of the event that a user cannot find its requested file in the caches of its neighbours. The provided simulations   verify the accuracy of the proposed analysis, and illustrate the effectiveness of the proposed NOMA based wireless caching schemes.  
\end{itemize}

The remainder  of the paper is organized as follows. In Section \ref{Section system model}, the considered system model, including the   caching model and the   spatial model, are introduced. In Section \ref{section pthen d}, the NOMA-assisted push-then-deliver strategy is presented, and its impact on the content pushing and delivery phases is investigated. In Section \ref{section cand d}, the proposed push-and-deliver strategy is developed  by efficiently  merging the content pushing and delivery  phases, its  impact   on the cache hit probability is investigated, and  its extension to D2D scenarios is   discussed. Computer simulations   are provided in Section \ref{section simulation}, and the paper is concluded in Section \ref{section conclusions}. The details of all  proofs are collected in the appendix. 

\section{System Model}\label{Section system model}
Consider a two-tier heterogeneous  communication scenario, where  multiple users request cacheable content with the help of   one   BS and multiple content servers. The D2D   scenario without caching infrastructure, e.g., content servers, will be described in Section \ref{section cand d}.C. Assume that each user is associated with a single content server. If the file requested by a user can be found in the cache of its associated content server, this server will serve the user, which means that   multiple content servers can communicate with their respective  users concurrently and hence  the spectral efficiency is high. However, if the file requested by a user cannot be found locally, the BS will serve the user directly, a situation which is not ideal for caching and should be avoided. The assumption that each user is associated with a single content server   facilitates the use of PCP modelling, as discussed in the following subsection.

\subsection{Spatial  Clustering  Model}
Assume  that the   BS is located in the origin of a two-dimensional Euclidean plane, denoted by $\mathbb{R}^2$.  As shown in Fig. \ref{fig01}, there are multiple content servers. The locations of the content servers  and the users are modelled as  PCPs. In particular, assume that the locations of the content servers  are denoted by $x_i$ and are modelled as a homogeneous Poisson point process (HPPP), denoted by $\Phi_c$, with density $\lambda_c$, i.e., $x_i\in \Phi_c$. For   notational simplicity, the   location of the BS is denoted by $x_0$. 

Each content server  is the parent node of a cluster covering a disk whose radius is denoted by  $\mathcal{R}_c$. Denote the content server  in cluster $i$ by $\text{CS}_i$. Without loss of generality, assume that there are $K$ users associated with $\text{CS}_i$,  denoted by $\text{U}_{i,k}$. Note that   users associated with the same content server  are viewed as offspring nodes \cite{Haenggi}. The   offspring nodes      are uniformly distributed in the disk associated with $\text{CS}_i$, and their locations are  denoted by $y_{i,k}$. To simplify the notation, the locations of the cluster users are conditioned on the locations of their cluster heads (content severs). As such, the distance from a   user to its content server   is simply given  by $||y_{i,k}||$, and the distance from user $\text{U}_{i,k}$ to content server  $\text{CS}_j$ is denoted by $||y_{i,k}+x_i-x_j||$ \cite{ 5560889, 7110502}.

\begin{figure}[!htbp]\centering\vspace{-1em}
    \epsfig{file=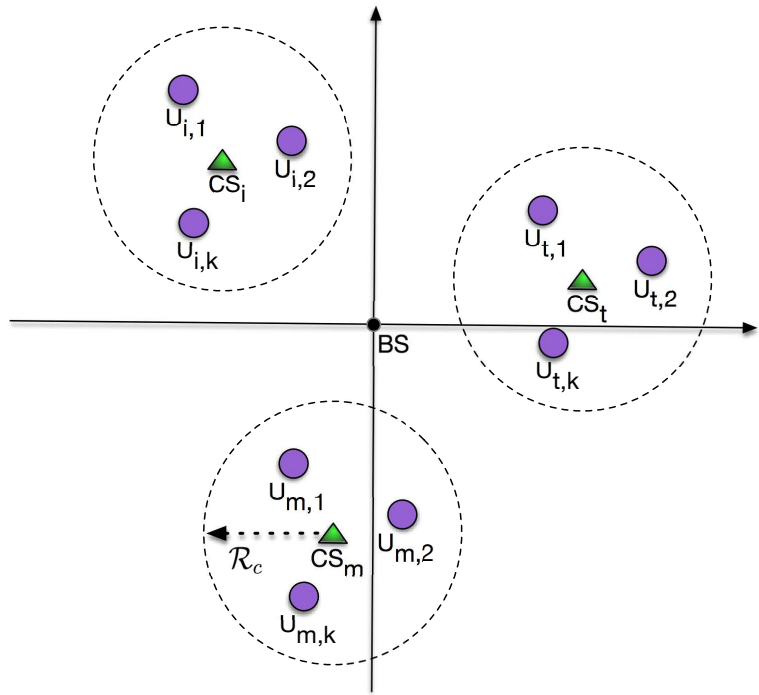, width=0.64\textwidth, clip=}\vspace{-1em}
\caption{ An illustration of the assumed spatial model. \vspace{-1em} }\label{fig01}
\end{figure}

\subsection{Caching Assumptions}
Consider that the files to be requested by the users are collected in a finite content library $\mathcal{F}=\{f_1, \cdots, f_F\}$. 
The popularity of the requested files is modelled by a Zipf distribution \cite{6600983}. Particularly,   the popularity of   file $f_i$, denoted by $\mathrm{P}(f_l)$,  is modelled  as follows:
\begin{align}\label{popular}
\mathrm{P}(f_l) = \frac{\frac{1}{l^{\gamma}}}{\sum^{F}_{p=1}\frac{1}{p^\gamma}},
\end{align}
where $\gamma>0$ denotes the shape parameter defining  the content popularity skewness. We note that $\mathrm{P}(f_l)$ is the probability  that a user   requests  file $f_i$. Similar to the existing   wireless caching literature, 
\cite{6871674,6495773,   7488289,7565184,7828114},   packets belonging to different files  are assumed to have the same length. However, unlike the existing literature, we do not  assume that the amount of information contained in  the  packets of different files  is identical.  Particularly, the prefixed data rate of  the  packets of file $f_l$ is denoted by   $R_l$.  We assume that    packets belonging to different files have the same size but may contain different amounts of information for the following reasons. {\it Firstly,}   the packet size is typically predefined according to practical system standards and cannot be changed.  Therefore, it is reasonable to   assume that all  packets have the same size. {\it Secondly}, packets belonging to different files have different priorities and different target  reception reliabilities, which requires  the use of different channel coding rates for  different packets. As a result,  packets which have the same size do not necessarily contain the same amount of information. We note that all  analytical results developed in this paper, except for Lemma \ref{lemma2}, do  not require  the assumption that the packets contain different amounts of information. In fact,   the performance for the special case of identical   target data rates for all   files is investigated   in the simulation section. 

Finally, we assume that the  BS has access to all   files. In this paper, when content servers exist,   we assume that the users have no caching capabilities.  On  the other hand, for the D2D assisted caching  discussed in Section \ref{section cand d}.C, it is assumed that each user has a cache.  

\section{ Push-then-deliver Strategy   } \label{section pthen d}

In this  section, we  consider the case where  the two caching phases, content pushing and content delivery, are separated. Unlike conventional caching which relies on the use of off-peak hours for content pushing,   the proposed push-then-deliver strategy assumes that   limited  bandwidth resources, such as  short time intervals, are periodically reserved for   pushing new files to the content servers  during on-peak hours. For example, every   hour, a BS deployed in a large  shopping mall or an airport may use  a few seconds  to push   updated advertising and marketing videos to the content servers.   The time interval  reserved for content pushing during on-peak hours has to be short in order to achieve  high system spectral efficiency. As will be shown,   the proposed push-then-deliver strategy allows  more files to be pushed within this short time interval compared to OMA.  

 In the following two subsections, we will demonstrate the impact of the NOMA principle on the content pushing and delivery phases, respectively. 

\begin{figure}[!htp]\vspace{-1em}
\begin{center} \subfigure[ A general illustration of the  content pushing phase ]{\label{fig 1x a}\includegraphics[width=0.49\textwidth]{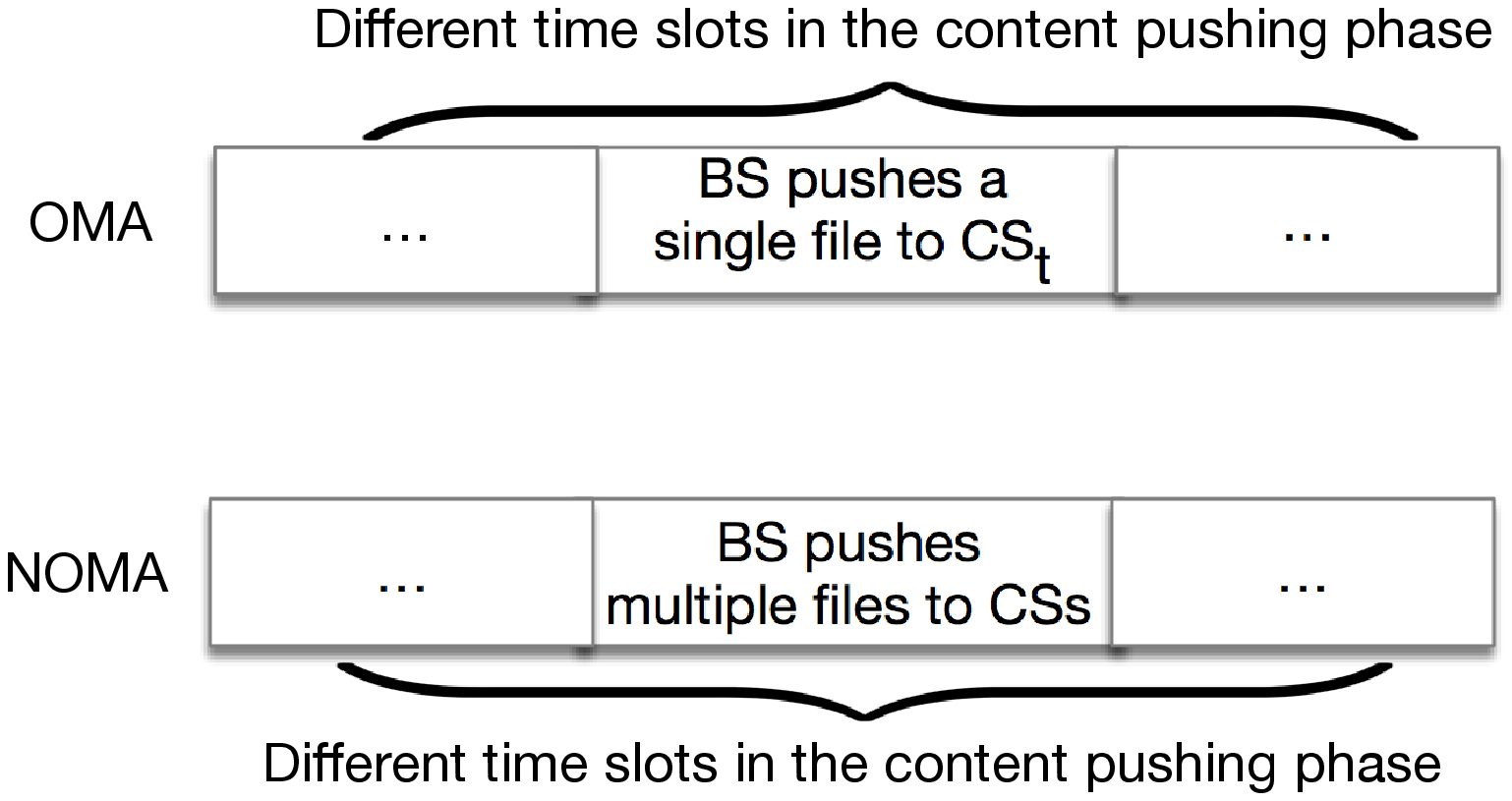}}
\subfigure[An example for the outcome of content pushing ]{\label{fig 1x b}\includegraphics[width=0.49\textwidth]{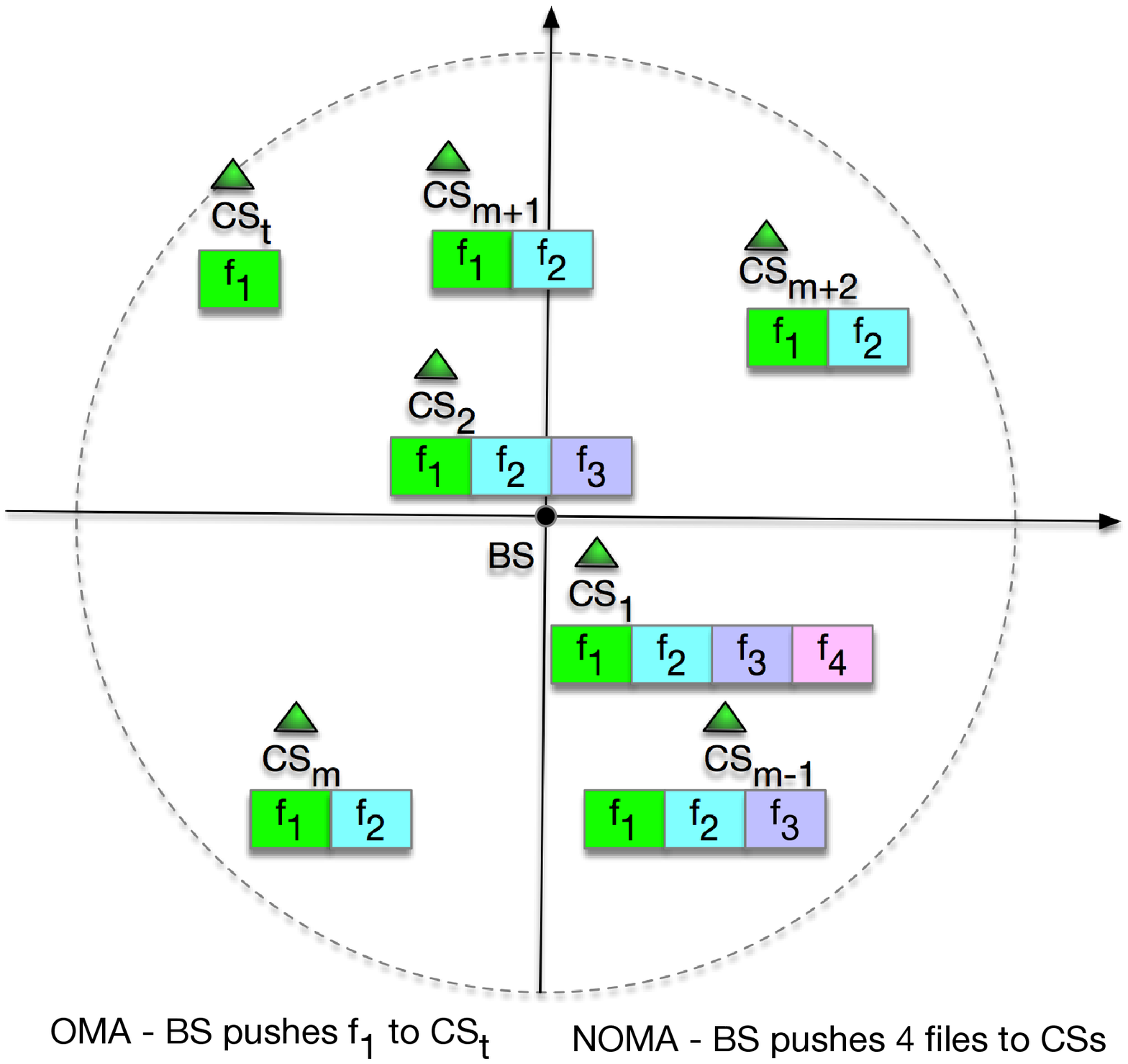}}\vspace{-1em}
\end{center}
\caption{  An illustration of the impact of NOMA on content pushing.  ``CS" denotes a content server in the subfigures.  For the example shown in  subfigure (b),  it is assumed that $\text{CS}_m$ is closer to the BS than  $\text{CS}_t$,  for $1< m< t $.   In OMA, a single file is pushed to $\text{CS}_t$, and in NOMA,   the BS pushes a superimposed mixture consisting of four  files, where  content servers closer to the BS are likely able to decode   more  pushed files.   }\label{fig01x}\vspace{-0.5em}
\end{figure}
  
\subsection{Content Pushing Phase}\label{subsection IIIA}
 In order to have a baseline for  the performance of NOMA assisted content pushing,   conventional OMA based content pushing   is introduced first. 
\subsubsection{OMA  Based Content Pushing}
 In OMA, the BS divides the content pushing phase into multiple time slots, and different content servers are served during different time slots,  as shown in Fig. \ref{fig 1x a}.     Without loss of generality, we focus on a time slot in which    the BS  needs to push a file to $\text{CS}_t$.  Since the use of OMA means  that only a single file can be pushed during this time slot, the BS will push the most popular file $f_1$ to $\text{CS}_t$\footnote{ We note that how  content is   pushed over the limited number of     time slots available for content pushing during on-peak hours  depends on the file scheduling strategy. One possible  strategy is to divide the popular content into different libraries or different sets, where the most popular files in these libraries or sets will be pushed during those  time slots. As a result, the broadcast nature of wireless channels  can be exploited, and   a   transmission  pushing  a file  to an intended content server  also benefits   other content servers which are     interested in the same content and have stronger connections to the BS than the intended content server.  An interesting   direction for future research is the design of   sophisticated algorithms for file   scheduling  during  the short time interval    available  for content pushing in on-peak hours.}   . 

Therefore,  $\text{CS}_t$   is able to   decode   file $f_1$ with the following achievable data rate: 
\begin{align}
R^{CP}_{t,OMA} = \log\left(1+\rho \frac{1}{{L\left(||x_t-x_0||\right)}}\right),
\end{align}
where $\rho$ denotes the transmit signal-to-noise ratio (SNR),  and $\frac{1}{L\left(||x_{t}||\right)}$ is the large scale path loss between $\text{CS}_t$ and the BS located at $x_0$. Particularly, the following path loss model is used, $\frac{1}{L\left(||x_{t}||\right)}$, where $L\left(||x_{t}||\right)=||x_{t}||^\alpha$ and  $\alpha$ denotes the path loss exponent. For a large scale  network, the probability   that  $||x_t-x_0||<1$ is very small, and therefore, the simplified unbounded path loss model is used in this paper \cite{4086349, 5560889, 7110502}. Nevertheless,  the presented  analytical results can be  extended to other path  loss models, e.g., $L\left(||x_{t}||\right)=||1+x_{t}||^\alpha$ or $L\left(||x_{t}||\right)=\max\{1, ||x_{t}||^\alpha\}$,  in a straightforward manner.   We note that small scale multi-path fading is not considered for the channel gain associated with $\text{CS}_t$ since the content servers can be deployed such that  line-of-sight connections to the BS are ensured, which means that large scale path loss is the dominant factor for    signal attenuation.  However, small scale fading   will be considered for the channel gains associated with the users, since  the users may not  have line-of-sight connections to their transmitters.

\subsubsection{NOMA Assisted Content Pushing} \label{subsection x}
By applying the concept of NOMA, more content can be simultaneously  delivered from the BS to the content servers, as illustrated  in Fig. \ref{fig 1x b}. Particularly,     the   BS sends the following mixture, which contains the $M_s$ most popular files:
\begin{align}
s_i = \sum^{M_s}_{i=1}\alpha_i \bar{f}_{i} ,
\end{align}
where $\bar{f}_{i}$ denotes the signal which represents the information contained in file $f_i$, $\alpha_i$ denotes the real-valued power allocation coefficient and $\sum^{M_s}_{i=1}\alpha_i^2=1$. The content servers carry out successive interference cancellation (SIC). The SIC decoding order is determined  by the priority of the files, i.e., a more popular file, $f_i$, will be decoded before a less popular file, $f_{j}$, $i<j$. Suppose that the files $f_j$, $j<i$, have been decoded and subtracted  correctly by content server $\text{CS}_m$. In this case,  $\text{CS}_m$ can decode the next most popular  file, $f_i$, with the following data rate:
\begin{align}
R^{CP}_{m,i} =&  \log\left(1+\frac{\rho \alpha_i^2\frac{1}{{L\left(||x_m-x_0||\right)}}}{\rho \frac{1}{{L\left(||x_m-x_0||\right)}}\sum^{M_s}_{j=i+1}\alpha_j^2+1}\right).
\end{align}
If $R^{CP}_{m,i}\geq R_i$,  then  file $f_i$ can be decoded and subtracted correctly at $\text{CS}_m$.
 
 For a fair  comparison  with   OMA,  which  pushes only one file  at a time, a sophisticated  power allocation policy is   needed for the NOMA scheme. Without loss of generality, we assume that the content servers are ordered as follows:
 \begin{align}
\frac{1}{{L\left(||x_1-x_0||\right)}}  \cdots \geq\frac{1}{{L\left(||x_m-x_0||\right)}} \geq \cdots \geq \frac{1}{{L\left(||x_t-x_0||\right)}}\geq \cdots ,
 \end{align}
 for $1\leq m < t$. Furthermore, since the considered time slot is used to push $f_1$ to $\text{CS}_t$, we make the following quality of service (QoS) assumption, in order to facilitate  the design of the power allocation coefficients:
 \[
 \textit{QoS   Target: The most popular file, $f_1$, needs to reach the $t$-th nearest content server ($\text{CS}_t$).} 
 \]
 Both the OMA and NOMA transmission schemes need to ensure this QoS target. Therefore, the  CR inspired power allocation policy can be used for NOMA \cite{Zhiguo_CRconoma}, i.e.,   power allocation coefficient  $\alpha_1$ is chosen such  that $f_1$ can be delivered reliably to  $\text{CS}_t$, i.e.,   
\begin{align}
R^{CP}_{t,1} \geq R_1.
\end{align}
 This constraint  results in the following choice of $\alpha_1$:
\begin{align}\label{c1}
\alpha_1^2=\min \left\{1,  \frac{\epsilon_1\left(\rho \frac{1}{{L\left(||x_t-x_0||\right)}} +1\right)}{\rho (1+\epsilon_1)\frac{1}{{L\left(||x_t-x_0||\right)}} }\right\},
\end{align}
where $\epsilon_l=2^{R_l}-1$. The use of the power allocation policy in \eqref{c1}   ensures that the outage probability for pushing file, $f_1$, to $\text{CS}_t$  is the same as that for OMA. The reason  is that if there is an outage in OMA,  $\alpha_1$ becomes one, i.e., all the power is allocated to $f_1$. Or in other words,  additional files are pushed in NOMA only if  $f_1$ is pushed to $\text{CS}_t$ successfully. 

Since $\sum^{{M_s}}_{j=1}\alpha_j^2=1$,   \eqref{c1}   implies that the sum of the power allocation coefficients, excluding $\alpha_1$, is constrained as follows: 
\begin{align}\label{c2}
\sum^{{M_s}}_{j=2}\alpha_j^2=\max \left\{0,  \frac{ \rho \frac{1}{{L\left(||x_t-x_0||\right)}}  -\epsilon_1}{\rho (1+\epsilon_1)\frac{1}{{L\left(||x_t-x_0||\right)}} }\right\}.
\end{align}
The constraint in \eqref{c1} is sufficient to guarantee the successful delivery of $f_1$ to the $t$-th nearest  content server.  How the remaining power shown in \eqref{c2} is allocated to the other files, $f_i$, $i\neq 1$, does not affect the delivery of $f_1$. Therefore, in this paper, it is assumed that the portion allocated to $f_i$, $i\neq 1$, is fixed, i.e., $\alpha_i^2=\beta_i P_r$, where $P_r=\max \left\{0,  \frac{ \rho \frac{1}{{L\left(||x_t-x_0||\right)}}  -\epsilon_1}{\rho (1+\epsilon_1)\frac{1}{{L\left(||x_t-x_0||\right)}} }\right\}$ and   $\beta_i $ are constants,  which satisfy the constraint $\sum^{{M_s}}_{j=2}\beta_i=1$. Note that the coefficients, $\beta_i $, indicate how the remaining transmission power, after   the power for  $f_1$ has been  deducted is allocated to the additional files. 

\subsubsection{Performance Analysis}
An important  criterion for evaluating   content pushing is  the cache hit probability which is the probability that, during the content delivery phase, a user finds its requested file in the cache of its associated content server\footnote{ We note that retransmission for content pushing, where after decoding failure, $\text{CS}_t$ requests the retransmission of file $f_1$ by the BS, is not considered in this paper. However, investigating  the impact of file retransmission on the caching performance is an important direction for future research. }. Since the request probability for file $l$ is decided by its popularity, the hit probability for a user associated with $\text{CS}_m$ can be expressed as follows:
\begin{align}\label{hit pr}
\mathrm{P}^{hit}_{m} = \sum^{M_s}_{i=1}\mathrm{P}(f_i) (1-\mathrm{P}_{m,i}),
\end{align}
where $\mathrm{P}_{m,i}$ denotes the outage probability of $\text{CS}_m$ for decoding  file $i$. Note that for the OMA case, only file $f_1$ will be sent, and hence the corresponding OMA hit probability is simply given by 
\begin{align}
\mathrm{P}^{hit}_{m,OMA} =  \mathrm{P}(f_1) (1-\mathrm{P}^{OMA}_{m,1}),
\end{align}
where $\mathrm{P}^{OMA}_{m,1}$ denotes the outage probability of $\text{CS}_m$ for decoding  file $f_1$. The following theorem reveals the benefit of using NOMA for content pushing. 

\begin{theorem}\label{theorem1}
The cache hit probability achieved by the proposed  NOMA assisted push-then-deliver strategy  is always larger than or at least equal to that of the conventional OMA based strategy, i.e., 
\begin{align}\label{theorem eq}
\mathrm{P}^{hit}_{m}\geq \mathrm{P}^{hit}_{m,OMA} ,
\end{align} 
for $1\leq m\leq t$. 
\end{theorem}

\begin{proof}
See  Appendix A. 
\end{proof}
 {\it Remark 1:} Theorem \ref{theorem1} clearly demonstrates   the superior performance of the proposed NOMA assisted caching strategy compared to  OMA based caching. This performance gain originates from  the fact that multiple content files are pushed concurrently during the content pushing phase.

{\it Remark 2:} Only the $t$ nearest content servers are of interest in \eqref{theorem eq}, i.e., $1\leq m\leq t$, which is due to our   assumption that the BS aims to push the most popular file, $f_1$,  to $\text{CS}_t$. 

{\it Remark 3:} As shown in Appendix A,  the key step to prove the theorem is to show     $\mathrm{P}^{OMA}_{m,1}=\mathrm{P}_{m,1}$, i.e., the   outage performance of NOMA for decoding $f_1$ at $\text{CS}_m$ is the same as that of OMA. If $f_1$ is viewed as the message to the primary user in a CR NOMA system, this observation about the equivalence between the  outage performances of   NOMA and OMA is consistent with the results  in \cite{Zhiguo_CRconoma} and \cite{Zhiguo_mimoconoma}. 

While the use of the CR power allocation policy guarantees that  $\text{CS}_t$ can decode $f_1$, this also implies  that the outage performance at $\text{CS}_m$ depends on the channel conditions of $\text{CS}_t$. This means that for   calculation of the outage probability, $\mathrm{P}_{m,i}$,    the joint distribution of the ordered distances  of  $\text{CS}_t$ and $\text{CS}_m$ to the BS is needed. The following lemma provides an analytical  expression for this joint distribution.  

\begin{lemma}\label{lemma1}
Denote the distance between the BS and the $i$-th nearest content server by $r_i$. 
The joint pdf of $r_m$ and $r_t$ is given by 
\begin{align}
f_{r_m,r_t}(x,y) =& 4y(\lambda_c\pi)^{t} e^{-\lambda_c \pi y^2}  \frac{  x^{2m-1}(y^2-x^2)^{t-m-1}}{(t-m-1)!(m-1)!}    . 
\end{align}
\end{lemma}
\begin{proof}
See Appendix B. 
\end{proof}
 {\it Remark 4:} We note that Lemma \ref{lemma1} is general and can be applied to any two HPPP nodes which are ordered according to their distances to the origin.  

{\it Remark 5:} It is worth pointing out that  the joint pdf obtained  in \cite{6909064} is  a special case of Lemma \ref{lemma1}, when    $m=1$ and $t=2$.

Since the cache hit probability is a function of the outage probability, we provide   the outage performance for content pushing   in the following lemma. 
\begin{lemma}\label{lemma2}
Assume $\epsilon_{M_s}\geq \epsilon_{1}$. The outage probability of $\text{CS}_n$, $1\leq n\leq t$, for decoding  $f_1$ is given by
\begin{align}
\mathrm{P}_{n,1}  =e^{-\lambda_c\pi \left(\frac{\rho}{\epsilon_1}\right)^{\frac{2}{\alpha}}}\sum^{n-1}_{k=0} \frac{(\lambda_c\pi)^{k}\left(\frac{\rho}{\epsilon_1}\right)^{\frac{2k}{\alpha}}}{k! }.
\end{align}

The outage probability of $\text{CS}_t$ for decoding  $f_i$, $2\leq i\leq M_s$, is given by
\begin{align}
\mathrm{P}_{t,i} & =e^{-\lambda_c\pi \left( \frac{\epsilon_1}{\rho} +\frac{(1+\epsilon_1)}{\rho \phi_i}\right)^{-\frac{2}{\alpha}}}\sum^{t-1}_{k=0} \frac{(\lambda_c\pi)^{k}\left( \frac{\epsilon_1}{\rho} +\frac{(1+\epsilon_1)}{\rho \phi_i}\right)^{-\frac{2k}{\alpha}}}{k! },
\end{align}
where $\phi_i=\min\left\{\frac{\bar{\xi}_2}{\epsilon_2}, \cdots, \frac{\bar{\xi}_i}{\epsilon_i}\right\}$,  $\bar{\xi}_i=\left(\beta_i -\epsilon_i \sum^{{M_s}}_{j=i+1}\beta_j  \right)$ for $2\leq i<M_s$, and $\bar{\xi}_{M_s}= \beta_{M_s}  $. 

The outage probability of $\text{CS}_m$, $1\leq m<t$, for decoding $f_i$, $2\leq i\leq M_s$, is given by
\begin{align}\nonumber 
\mathrm{P}_{m,i}&\approx  \mathrm{P}_{t,1}+\frac{ 4(\lambda_c\pi)^{t}}{(t-m-1)!(m-1)!}    \sum^{t-m-1}_{p=0}(-1)^p{t-m-1 \choose p} \\\nonumber &\times  \sum^{N}_{l=1}\frac{\pi\left(\tau_2-\tau_1\right)}{2 N} f_m\left(\frac{\tau_2-\tau_1}{2}w_l+\frac{\tau_2+\tau_1}{2}\right)\sqrt{1-w_l^2},
\end{align}
where $\tau_1=\left(\frac{\rho \phi_i}{1+\epsilon_1+\epsilon_1\phi_i}\right)^{\frac{1}{\alpha}}$, $\tau_2=\left(\frac{\epsilon_1}{\rho }\right)^{-\frac{1}{\alpha}}$, $N$ denotes the parameter for Chebyshev-Gauss quadrature,  $w_l=\cos\left(\frac{2l-1}{2N}\pi\right)$, $g(y)=\left(\frac{(1+\epsilon_1)}{ \phi_i \left( \rho   -\epsilon_1 y^{\alpha}\right)}\right)^{-\frac{1}{\alpha}}$, and
\begin{align} 
f_m(y) =& \frac{e^{-\lambda_c \pi y^2} y^{2(t-m-1)-2p+1}    }{2m+2p}  \left(y^{2m+2p} - (g(y))^{2m+2p})\right). 
\end{align}  
\end{lemma}

\begin{proof}
See Appendix C. 
\end{proof}
 {\it Remark 6:} By using the closed-form expressions  developed in Lemma \ref{lemma2}, the outage probabilities for content pushing can be  evaluated   in a straightforward manner, and  computationally challenging  Monte Carlo simulations can be avoided.    

{\it Remark 7:} We note that,  in Lemma \ref{lemma2} it is assumed that the target data rates and the power allocation coefficients are chosen to ensure $\bar{\xi}_i>0$. Otherwise, an outage will always happen  for decoding     file  $f_i$, $i\geq 2$, at the content servers. 

{\it Remark 8:} In  Lemma \ref{lemma2}, it is also assumed that $\epsilon_{M_s}\geq \epsilon_{1}$, in order to avoid a trivial case for the integral calculation, as shown in \eqref{termx1} in Appendix C. This assumption means that the target data rate for file $f_{M_s}$ is larger than that of file $f_1$, and the assumption is important to the performance gain of the NOMA assisted strategy over the OMA based one, as explained in the following.  Recall that   the CR power allocation policy   in \eqref{c1}   ensures that $\text{CS}_t$ can decode the pushed file $f_1$. If   there is any power left after satisfying the needs of $\text{CS}_t$, the BS will use the remaining power to push additional  files. If the target data rate of $f_1$, $R_1$, is very large, meeting  the decoding requirement of $\text{CS}_t$ becomes challenging, i.e., it is very likely that there is not much power   available for pushing additional files. In this case, the use of the proposed push-then-deliver strategy will not offer much performance gain compared to   OMA.  In other words,  a large $R_1$ is not   ideal   for applying  the proposed NOMA based   push-then-deliver strategy. The proposed strategy  is more beneficial in scenarios when the  target data rate of $f_1$ is small.

\subsection{Content Delivery Phase} 
In the previous subsection, the cache hit probability for content delivery has been analyzed. However, the event that  a user can find its requested file in the cache of its associated content server is not equivalent to the event  that this user can receive the file correctly, due to  the multi-path fading and path loss attenuation that affect its link to the content server.  Hence, in this subsection, the impact of NOMA on the reliability of content delivery is investigated. Similar to the previous subsection, the conventional OMA based content delivery strategy  is described first as a benchmark scheme. 
\subsubsection{OMA Based Content Delivery}
 Similar to the content pushing phase, the content delivery phase is also divided into multiple time slots, as shown in Fig. \ref{fig 12 a}.  During each time slot,  for the OMA case,  each content server  randomly schedules a single user whose requested file  is available   in its cache. We assume that each content server can find a user to serve for this OMA based content delivery, and multiple content servers    help their associated users simultaneously. 

\begin{figure}[!htp]\vspace{-1em}
\begin{center} \subfigure[ A general illustration for the  content delivery  phase ]{\label{fig 12 a}\includegraphics[width=0.49\textwidth]{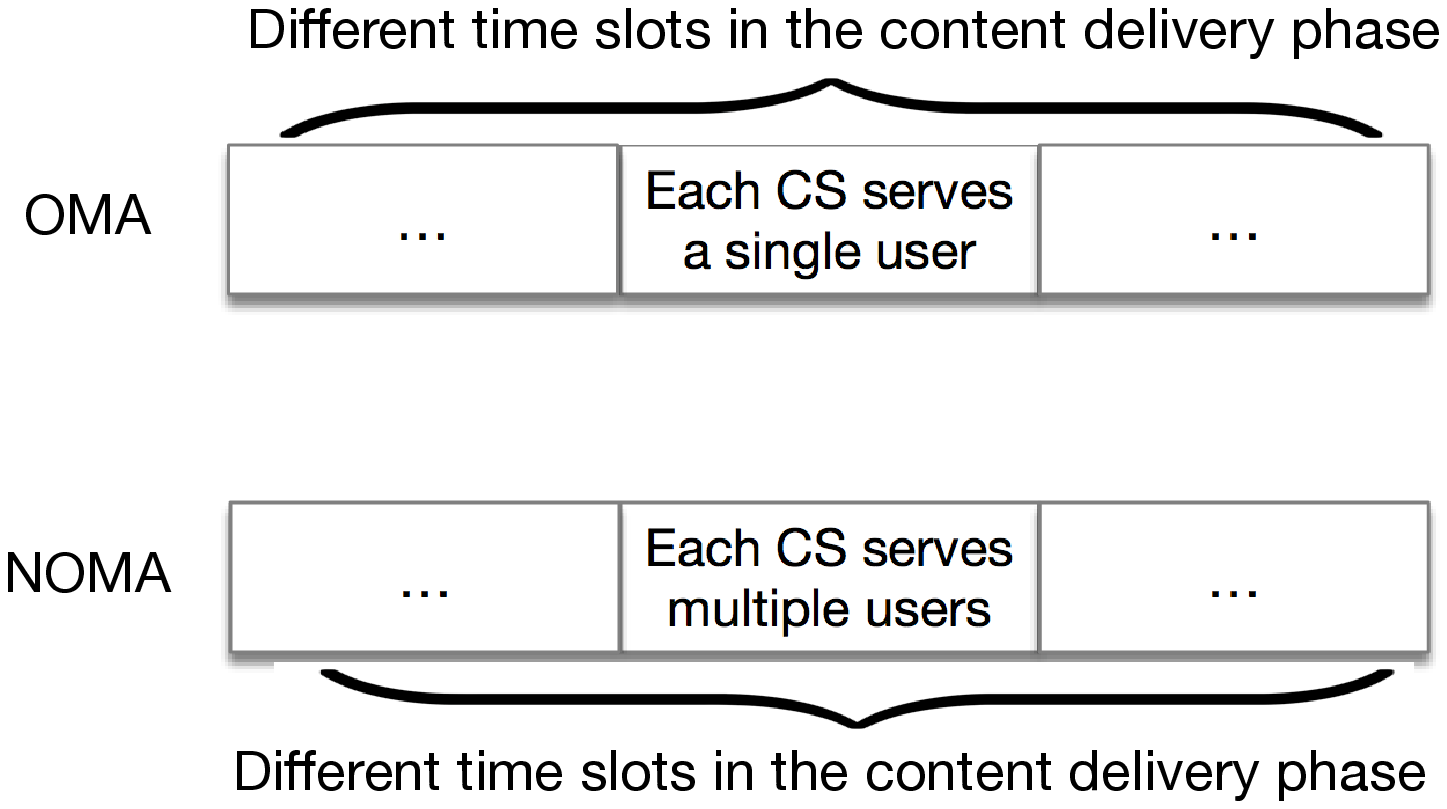}}
\subfigure[An example for the outcome of content pushing ]{\label{fig 1 b}\includegraphics[width=0.49\textwidth]{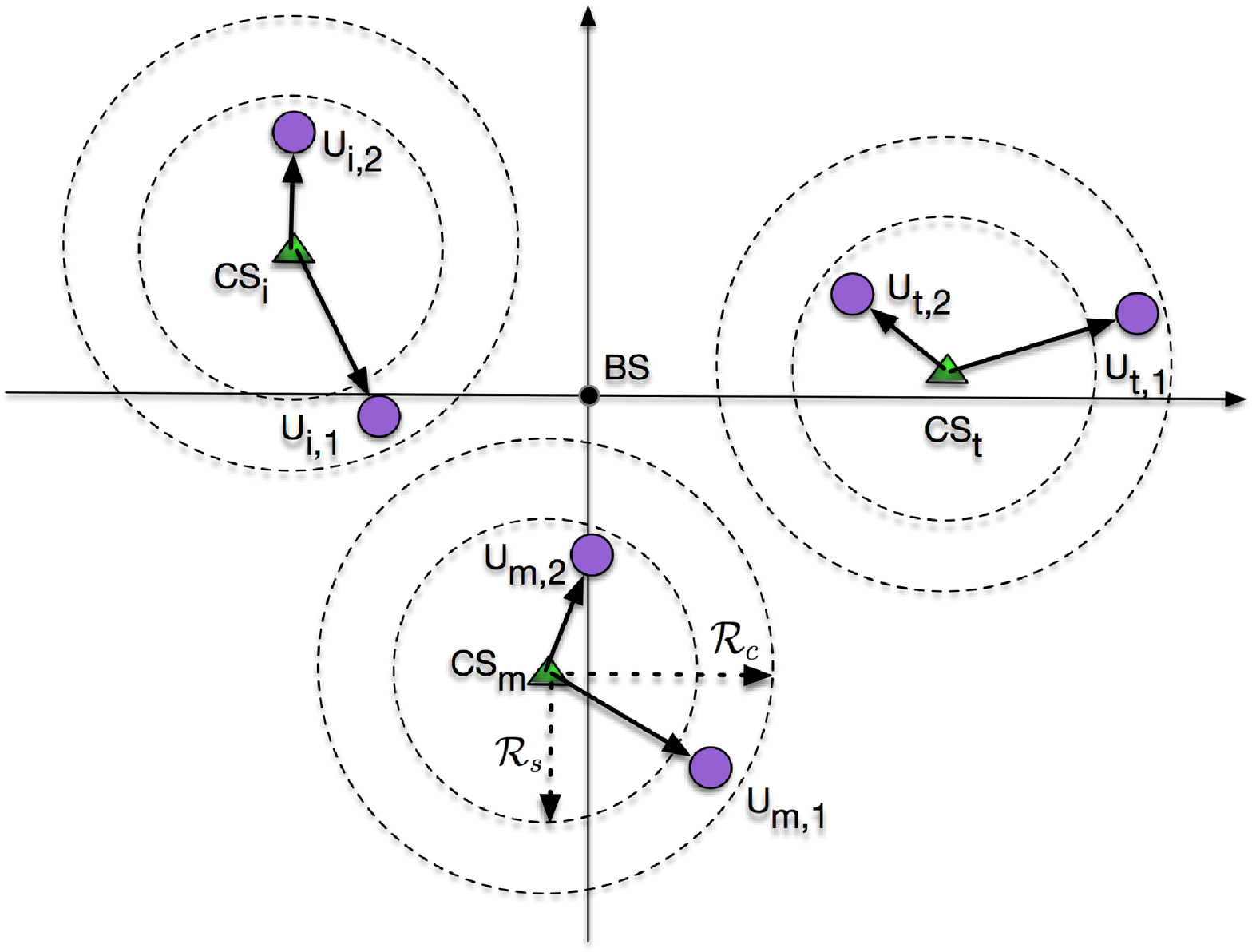}}\vspace{-1em}
\end{center}
\caption{  An illustration of the impact of NOMA on content delivery.   In OMA, each content server serves a single user. By using NOMA, an additional user can be served.     }\label{2fig02b2}\vspace{-0.5em}
\end{figure}

\subsubsection{NOMA Assisted Content Delivery}
If   the NOMA principle is applied in the content delivery phase, each content server can serve two users\footnote{ We focus on   the case with two users  since the content delivery phase is analog to the conventional downlink case and two-user NOMA based downlink transmission has been proposed for   long term evolution (LTE) Advanced \cite{3gpp1}. The analytical results presented  in this paper can be  extended to the case with more than two users by dividing the disc covered by a content server into multiple   rings. In practice,   the  number of users to be served simultaneously needs to reflect a practical tradeoff between system complexity and throughput. }.   Thereby, it is assumed that each content server can find at least two users  whose requests can be accommodated locally\footnote{ Please note that we did not make any assumption regarding  which particular user   makes a request to its content server. This is the reason why the locations and fading channels of the users are modelled as random.}. This assumption is applicable to high-density wireless networks,  such as networks deployed in sport stadiums or airports, where the number of users is much larger than the number  of content servers. We note that this assumption constitutes  the worst case for   the reception reliability of the users.  In fact,  content servers that  do not have any user to server   will not cause interference to  the users served by  other content servers.   Without loss of generality, assume that the two users are ordered based on their distances to the associated content server. As shown in Fig. \ref{2fig02b2},  the far user, which is far from the content server and is denoted by  $\text{U}_{m,1}$,  is inside  a ring bounded by  radii   $\mathcal{R}_s$ and $\mathcal{R}_c$, $\mathcal{R}_s<\mathcal{R}_c$. The near user, which is close  to the content server and is denoted by $\text{U}_{m,2}$, is inside  a disc with radius $\mathcal{R}_s$. Without loss of generality, denote the file requested by $\text{U}_{m,k}$ by $f_{m,k}$, $f_{m,k}\in\mathcal{F}$. Each content server broadcasts a superposition  signal  containing two messages, and $\text{U}_{m,k}$, which is associated with $\text{CS}_m$,  receives the following signal: 
\begin{align}\label{model1}
y_{m,k} =& \underset{\text{Signals from   $\text{CS}_m$}}{\underbrace{ \frac{h_{m,mk}}{\sqrt{L\left(||y_{m,k}||\right)}} \sum^{2}_{l=1} \alpha_l\bar{f}_{m,l}}}   \\ \nonumber &+   \underset{\text{Signals from interfering  clusters}}{\underbrace{\sum_{x_j\in \Phi_c\backslash m}\frac{ h_{j,mk}}{\sqrt{L\left(||y_{m,k}+x_m-x_j||\right)}}  \sum^{2}_{l=1} \alpha_l  \bar{f}_{j,l}}}  +n_{m,k},
\end{align}
where $\bar{f}_{j,l}$ denotes the signal which represents the information contained in file $ {f}_{j,l}$, $\alpha_l$ denotes the NOMA power allocation coefficient,  $n_{m,k}$ is the additive complex Gaussian noise,  and $h_{j,mk}$ denotes the Rayleigh fading channel coefficient between $\text{CS}_j$ and $\text{U}_{m,k}$.   In order to obtain tractable analytical results, fixed power allocation is used, instead of  CR power allocation, and   it is assumed that all   content servers  use the same fixed power allocation coefficients. In order to keep the notations consistent, the power allocation coefficients are still denoted by $\alpha_i$.   We note that the simulation results provided in Section \ref{section simulation} show that the use of this fixed power allocation can still ensure that NOMA outperforms OMA for both   users. 

    $\text{U}_{m,1}$ will treat its partner's message as noise and decode its own message $f_{m,1}$ with the following SINR:
\begin{align}\label{rate1}
\text{SINR}^1_{m,1} =&  \frac{   \frac{ \alpha_1^2|h_{m,m1}|^2}{L\left(||y_{m,1}||\right)}  }{   \frac{ \alpha_2^2|h_{m,m1}|^2}{L\left(||y_{m,1}||\right)} + \text{I}^{m,1}_{inter}   +\frac{1}{\rho}}  ,
\end{align}
where 
\[
\text{I}^{m,1}_{inter}=\underset{x_j\in \Phi_c\backslash m}{\sum}\frac{  | h_{j,m1}|^2}{L\left(||y_{m,1}+x_m-x_j||\right)}.
\]
In practice, the content servers are expected to  use less transmission power than the BS, but for notational simplicity, $\rho$ is still used to denote the ratio between the transmission power of the content servers and the noise power. In Section \ref{section simulation}, for the presented computer simulation results,   different transmission powers are adopted for  the BS and the content servers. 

The near user,  $\text{U}_{m,2}$, intends to first decode its partner's message with the data rate $\log(1+\text{SINR}^1_{m,2})$, where $\text{SINR}^1_{m,2} $ is defined similarly to $\text{SINR}^1_{m,1}$, i.e., $\text{SINR}^1_{m,2}=  \frac{   \frac{ \alpha_1^2|h_{m,m2}|^2}{L\left(||y_{m,2}||\right)}  }{   \frac{ \alpha_2^2|h_{m,m2}|^2}{L\left(||y_{m,2}||\right)} + \text{I}^{m,2}_{inter}   +\frac{1}{\rho}} $, and  the inter-cluster interference, $ \text{I}^{m,2}_{inter}$,  is defined similarly to $ \text{I}^{m,1}_{inter}$. If $\log(1+\text{SINR}^1_{m,2})>R_1$, i.e., $\text{U}_{m,2}$ can decode its partner's message  successfully,  $\text{U}_{m,2}$ will remove $f_{m,1}$ and decode its own message with the following SINR:
\begin{align}\label{rate2}
\text{SINR}^2_{m,2} =&  \frac{   \frac{ \alpha_2^2|h_{m,m2}|^2}{L\left(||y_{m,2}||\right)}  }{   \text{I}^{m,2}_{inter}  +\frac{1}{\rho}}   .
\end{align}

  The outage probabilities of the two users are defined as follows: 
  \begin{align}
  \mathrm{P}^1_{m,1} = \mathrm{P} (\log(1+\text{SINR}^1_{m,1})<R_1),
  \end{align}
   and 
   \begin{align}
   \mathrm{P}^2_{m,2} = 1 - \mathrm{P} (\log(1+\text{SINR}^1_{m,2})>R_1,\log(1+\text{SINR}^2_{m,2})>R_2).
   \end{align}  The following lemma provides    closed-form expressions for these outage probabilities. 
 \begin{lemma}\label{lemma3}
 The outage probability of $\text{U}_{m,2}$ can be expressed as follows:
 \begin{align}\label{eqxc3}
\mathrm{P}^o_{m,2}  \approx&     1- \sum^{N}_{n=1}\bar{w}_ne^{-\frac{c_{n,\mathcal{R}_s}\frac{1}{\rho}}{\tilde{\tau}}}    q\left(\frac{c_{n,\mathcal{R}_s}  }{\tilde{\tau}}\right)  ,
 \end{align} 
 where $\tilde{\tau}=\min\left\{ \frac{\alpha^2_1-\epsilon_1\alpha^2_2}{\epsilon_1} ,  \frac{ \alpha^2_2 }{\epsilon_2}\right\}$,  
 $q(s)   = 
{\rm exp}\left(-   2\pi\lambda_{c} \frac{s^{\frac{2}{\alpha}}}{\alpha}\text{B}\left(\frac{2}{\alpha}, \frac{\alpha-2}{\alpha}\right) \right)$, $\text{B}(\cdot, \cdot)$ denotes the Beta function, 
 $\bar{w}_n=  \frac{\pi}{2N}\sqrt{1-w_n^2}
\left(w_n+1\right)$, $w_n$ is defined in Lemma \ref{lemma2}, and $c_{n,r}=\left(\frac{r}{2}w_n+\frac{r}{2}\right)^\alpha$.

 The outage probability of $\text{U}_{m,1}$ can be expressed as follows:
\begin{align} 
\mathrm{P}^o_{m,1}   \approx& 1+\frac{\mathcal{R}_s^2}{\mathcal{R}_c^2-\mathcal{R}_s^2} \sum^{N}_{n=1}\bar{w}_ne^{-\frac{c_{n,\mathcal{R}_s}\frac{\epsilon_1}{\rho}}{\alpha^2_1-\epsilon_1\alpha^2_2}}  q\left(e^{-\frac{c_{n,\mathcal{R}_s}\epsilon_1  }{\alpha^2_1-\epsilon_1\alpha^2_2}}  \right) \\\nonumber &- \frac{\mathcal{R}_c^2}{\mathcal{R}_c^2-\mathcal{R}_s^2} \sum^{N}_{n=1}\bar{w}_ne^{-\frac{c_{n,\mathcal{R}_c}\frac{\epsilon_1}{\rho}}{\alpha^2_1-\epsilon_1\alpha^2_2}}   q\left( e^{-\frac{c_{n,\mathcal{R}_c}\epsilon_1 }{\alpha^2_1-\epsilon_1\alpha^2_2}}  \right).
 \end{align}
 
 \end{lemma} 
 \begin{proof}
 See  Appendix D. 
 \end{proof}
   {\it Remark 9:} In the previous subsection, the CR power allocation policy is used and this type of power allocation   ensures that   the NOMA outage performance of the far user, $\text{U}_{m,1}$, is the same as  that for OMA. Since  fixed power allocation coefficients are used in this subsection  for content pushing, the performance of the far user is no longer guaranteed, but surprisingly,  our simulation results indicate that the use of NOMA can still yield    an    outage performance gain for the far user, compared to OMA, as shown in Section \ref{section simulation}.
   
\section{Push-and-deliver Strategy}\label{section cand d}
 While  the proposed push-then-deliver strategy   ensures that more files can be pushed to the content servers during a short content pushing phase, it still relies on the same principle as conventional caching in the sense that  content pushing and content delivery are separately  carried out.   Therefore, if  the time interval  between two adjacent content pushing phases  is large,  the caches of the content servers can be updated only infrequently. If new content arrives during the content delivery phase, the use of both conventional caching  and the proposed push-then-deliver strategy means that the base station has to wait   until  the next content pushing phase in order to update the caches of the content servers.   In contrast, the proposed push-and-deliver strategy provides  an efficient mechanism for frequently  updating   the files cached at the content servers  by exploiting opportunities    for content pushing during the content delivery phase as illustrated  in Fig. \ref{3fig01b1}.  In particular,  such opportunities arise,  when the BS   has to  serve a user directly during the content delivery phase,  as explained in the following.

 \begin{figure}[!htp]\vspace{-1em}
\begin{center} \subfigure[ General principle of push-and-deliver ]{\label{3fig01b1}\includegraphics[width=0.8\textwidth]{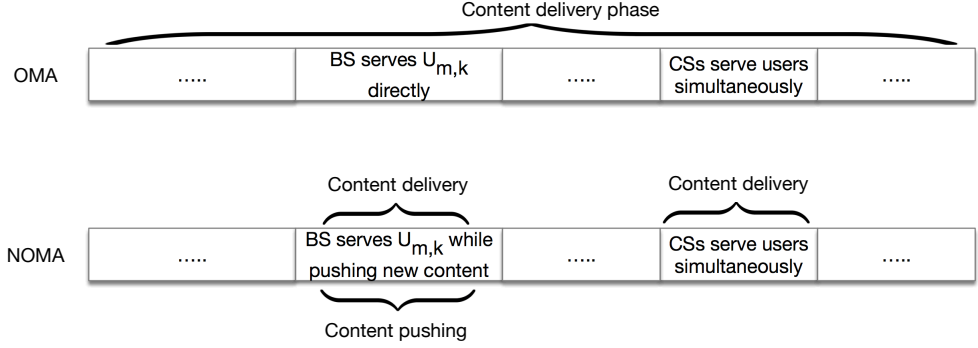}}
\subfigure[An   illustration of push-and-deliver]{\label{3fig02b2}\includegraphics[width=0.46\textwidth]{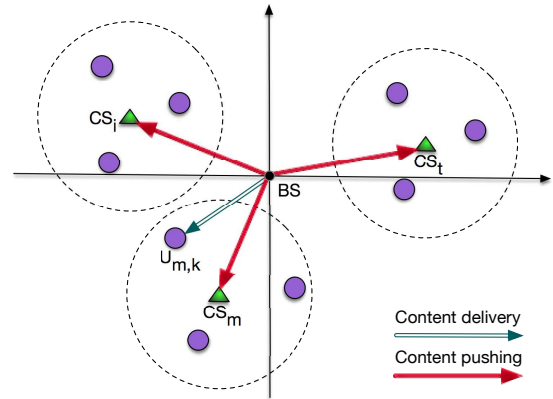}} 
\end{center}
\caption{  An illustration of the proposed push-and-deliver strategy.      }\label{fig04} 
\end{figure} 
In particular, consider a time slot which is dedicated to user $\text{U}_{m,k}$, as shown in Fig.  \ref{3fig02b2}. During this time slot, if  OMA  is used,  only this user can be served by the BS directly. However, the use of the NOMA principle offers  the opportunity to also push new content to the content servers\footnote{ We note that the BS cannot apply the proposed push-and-deliver  strategy in each time slot, as the need to serve a user directly is a prerequisite. Particularly, the proposed   strategy is applied in  time slots in which  a user cannot find the requested file at its local content server. More specifically, in these time slots, the BS pushes new content to the content servers, while severing the user directly. In other words, the push-and-deliver strategy can be applied whenever a user has to be served directly by the BS.   }, i.e.,  the BS sends a superposition signal  containing the file requested by $\text{U}_{m,k}$, denoted by ${f}_0$,  and the $M_s$ most popular files pushed by the BS, denoted by $f_i$, $1\leq i \leq M_s$.   Assume that  $ {f}_0$ and $f_i$, $1\leq i \leq M_s$, belong to different  file libraries, and the newly pushed files are useful to all the content servers, in order to avoid   correlation among these files and to simplify the expression for the cache hit probability.   This assumption is reasonable in practice and  can be justified by using the following scenario as an example.    Assume  that after a   period following the initial content pushing phase,   one user needs to be served by the BS directly, and during this       period,     a new library of files, e.g., a set of videos for breaking news, has arrived  at the BS to  be pushed to the content servers.  None of the content servers has had a  chance to cache these files yet. With the proposed push-and-deliver strategy, the BS  can push the new files to the content servers, while serving the requesting user directly, i.e.,   waiting  for the next content pushing phase is avoided.  

\subsection{Performance Analysis}\label{subs detl}
Following  similar steps  as in the previous section, the data rate of   $\text{U}_{m,k}$ for decoding its requested file, $f_0$, which is directly sent by the BS, is given by
\begin{align}\label{rate1}
\text{R}_{m,k} =& \log\left(1+ \frac{   \frac{ \alpha_0^2|h_{mk}|^2}{L\left(||y_{m,k}+x_m||\right)}  }{  \sum^{M_s}_{l=1} \frac{ \alpha_l^2|h_{mk}|^2}{L\left(||y_{m,k}+x_m||\right)}   +\frac{1}{\rho}}  \right),
\end{align}
and each content server, $\text{CS}_m$, can decode the additionally pushed file $f_i$ with the following data rate:
\begin{align}\label{rate1}
\text{R}^l_{m} =& \log\left(1+ \frac{   \frac{ \alpha_l^2}{L\left(|| x_m||\right)}  }{  \sum^{M_s}_{l=i+1} \frac{ \alpha_l^2 }{L\left(||x_m||\right)}   +\frac{1}{\rho}}  \right),
\end{align}
if $\text{R}^j_{m} $ is larger than $R_j$, for $0\leq j \leq i-1$, where  $R_l$ denotes the target data rate of $f_l$.  Again, small scale multi-path fading is not considered in the channel model for $\text{CS}_m$, as we assume that the large scale path loss is dominant in this case, but small scale fading is considered for the users' channels.  Note that the indices of the power allocation coefficients $\alpha_i$ start from $0$,  due to file $f_0$.  Compared to the distance between  $\text{CS}_m$ and the BS, the corresponding distance between $\text{U}_{m,k}$ and the BS has a very complicated pdf, as shown in the following subsection. Therefore, in order to obtain tractable analytical results,   fixed power allocation coefficients $\alpha_i$ will be used, instead of   the CR based  ones. The outage probabilities of the user and the content servers will be studied in the following subsections, respectively.

\subsubsection{Performance of the user}
The main challenge in analyzing the outage performance at the user is    the complicated expression for the pdf of the distance $||y_{m,k}+x_m||$. 
First, we define $ \bar{z}_{m,k}=   \frac{  |h_{mk}|^2}{L\left(||y_{m,k}+x_m||\right)} $.   The outage probability at the user can be expressed as follows:
\begin{align}
\mathrm{P}_{m,k}^1 = &\mathrm{P}(\text{R}_{m,k} <R_0)=\mathrm{P}\left(\bar{z}_{m,k}<\frac{{\epsilon}_0}{\rho \zeta_1}\right)\\\nonumber =& \mathcal{E}_{L\left(||y_{m,k}+x_m||\right)}\left\{1- e^{-L\left(||y_{m,k}+x_m||\right)  \frac{{\epsilon}_0}{\rho \zeta_0}}\right\} ,
\end{align} 
where  $\zeta_l=\alpha^2_l- {\epsilon}_l\sum^{M_s}_{j=l+1}\alpha^2_j$ for $0\leq l<M_s$, and $\zeta_{M_s}=\alpha^2_{M_s}$. Again, it is assumed that the power allocation coefficients and the target data rates are carefully chosen to ensure that $\zeta_l$ is positive. 

In order to derive  the pdf of $||y_{m,k}+x_m||$, we first define $r_m = ||x_m||$ and also a function $$g(r_m,r)= \frac{2r\arccos \frac{ r_m^2+r^2-\mathcal{R}_c^2 }{2r_m r}}{\pi \mathcal{R}_c^2}. $$  Conditioned on $r_m$,  the pdf of   $ ||y_{m,k}+x_m||  $ is given by \cite{7997055}
\begin{eqnarray}\label{rxm}
f_{||y_{m,k}+x_m|| } (r|r_m) = g(r_m,r)
,  \end{eqnarray}
for $ r_m- \mathcal{R}_c\leq r\leq r_m+ \mathcal{R}_c$, if $r_m>\mathcal{R}_c$. Otherwise, we have
\begin{eqnarray}\begin{array}{ll}
&f _{||y_{m,k}+x_m|| }(r|r_m) \\ = &\left\{\begin{array}{ll} 2\pi r, &\text{if } r\leq \mathcal{R}_c-r_m\\ 2\pi r-g(r_m,r), &\text{if }  \mathcal{R}_c-r_m <r\leq \sqrt{ \mathcal{R}_c^2-r_m^2}\\ g(r_m,r),&\text{if } \sqrt{ \mathcal{R}_c^2-r_m^2}<r\leq \mathcal{R}_c+r_m\end{array}\right..\end{array}
\end{eqnarray}

In order to avoid the trivial cases, which lead to $r=0$, i.e., the user is located at the same place as the BS, we assume that no   content server can be located  inside  the disc, denoted by $\mathcal{B}(x_0, \delta \mathcal{R}_c)$, i.e., a disc with the BS located  at its origin and   radius     $\delta \mathcal{R}_c$ with $\delta>1$, which means that  $r_m\geq \delta \mathcal{R}_c$  for all $m\geq 1$. Therefore, only the expression in \eqref{rxm} needs to be used since  $r_m$ is strictly larger than $\mathcal{R}_c$.

After using  the pdf of $ ||y_{m,k}+x_m||$, the outage probability can be expressed as follows:
\begin{align}\label{pdpmk} 
\mathrm{P}_{m,k}^1  =&    1- \int_{\delta\mathcal{R}_c}^{\infty}\int^{ z+ \mathcal{R}_c}_{z- \mathcal{R}_c}e^{-  \frac{{\epsilon}_0r^{\alpha}}{\rho\zeta_0}}  g(z,r)  d r \bar{f}_{r_m}(z)dz,
 \end{align} 
 where $\bar{f}_{r_m}(z)$ denotes the pdf of $r_m$.  
Because of the assumption that no content server can be located inside of $\mathcal{B}(x_0, \delta \mathcal{R}_c)$, the pdf of $r_m$ is different from that  in \eqref{ordered pdf}, but the steps of the proof for Theorem 1 in  \cite{1512427} can   still be applied to obtain the pdf of $r_m$.  Particularly, first denote   by $\mathcal{A}_r$ the ring with inner radius $\delta \mathcal{R}_c$   and outer radius $r$.  The cumulative distribution function (CDF)  of $r_m$ can be expressed as follows:
 \begin{align}
 \tilde{F}_{r_m}(r) =& 1 - \mathrm{P}(\#\text{ of nodes in the ring } \mathcal{A}_r<m)\\\nonumber 
  =&1 - \sum^{m-1}_{l=0}\frac{(\lambda_c[\pi r^2 - \pi\delta^2 \mathcal{R}_c^2])^l}{l!}e^{-\lambda_c[\pi r^2 - \pi\delta^2 \mathcal{R}_c^2]}.
 \end{align}
 Therefore, the pdf of $r_m$ can be calculated as follows: 
 \begin{align}\label{pdfx1}
 \bar{f}_{r_m}(r)  \nonumber
=& - 2\pi\lambda_c re^{-\lambda_c[\pi r^2 - \pi\delta^2 \mathcal{R}_c^2]}\left(\sum^{m-1}_{l=1}\frac{ (\lambda_c[\pi r^2 - \pi\delta^2 \mathcal{R}_c^2])^{l-1}}{(l-1)!}\right.\\\nonumber &\left. -\sum^{m-1}_{l=0}\frac{ (\lambda[\pi r^2 - \pi\delta^2 \mathcal{R}_c^2])^l}{l!} \right)
    \\
    =&  2\pi\lambda_c^m re^{-\lambda_c[\pi r^2 - \pi\delta^2 \mathcal{R}_c^2]} \frac{ [\pi r^2 - \pi\delta^2 \mathcal{R}_c^2]^{m-1}}{(m-1)!}  .  
 \end{align}
Substituting \eqref{pdfx1} into \eqref{pdpmk}, the outage probability of the user can be obtained.

\subsubsection{Performance of the content servers}
The content servers  need to carry out SIC in order to decode the newly pushed files $f_l$. As a result, the outage probability of $\text{CS}_m$ for decoding $f_i$  can be expressed as follows:
 \begin{align}
\mathrm{P}_{m}^i = &1 - \mathrm{P}(\text{R}_{m}^l >R_l,\forall l \in\{0, \cdots, i\}) \\\nonumber =&  
\mathrm{P}\left( L(||x_m||)>\min\left\{\frac{\rho \zeta_{l}}{{\epsilon}_l}, \forall l \in\{0, \cdots, i\}\right\} \right).
 \end{align}
 
 By applying the assumption that $r_m\geq \delta \mathcal{R}_c$ and also the pdf in \eqref{pdfx1}, the outage probability of  $\text{CS}_m$ for decoding $f_i$ can be expressed as follows:
  \begin{align}
\mathrm{P}_{m}^i = & \sum^{m-1}_{l=0}\frac{(\lambda_c\left[\frac{\pi}{\bar{\tau}_i^2 } - \pi\delta^2 \mathcal{R}_c^2\right])^l}{l!}e^{-\lambda_c\left[\frac{\pi}{\bar{\tau}_i^2}  - \pi\delta^2 \mathcal{R}_c^2\right]},
 \end{align}
 where $\bar{\tau}_i=\left(\frac{1}{\min\left\{\frac{\rho \zeta_{l}}{{\epsilon}_l}, \forall l \in\{0, \cdots, i\}\right\}} \right)^{\frac{1}{\alpha}}$. 
 
 Based on the outage probability $\mathrm{P}_{m}^i $,  the corresponding cache hit probability for a user associated with $\text{CS}_m$ can be   expressed as follows:
\begin{align}
\mathrm{P}^{hit}_{m} = \sum^{M_s}_{i=1}\mathrm{P}(f_i) (1-\mathrm{P}_{m}^i),
\end{align}
where $f_0$ has been omitted as it is the	 file currently requested by a user and  it is assumed that $f_0$ and $f_l$, $1\leq l\leq M_s$,  belong to  different libraries.

 \subsection{OMA Benchmarks}
 A naive OMA based benchmark is that the BS does not push new content while serving a user directly. Compared to this naive OMA scheme, the benefit of the proposed push-and-deliver strategy is obvious since new content is delivered and the cache hit probability will be improved. 
 
 A more sophisticated OMA scheme is to divide a single time slot into $(M_s+1)$ sub-slots. During the first sub-slot, the user is served directly by the BS. From the second until  the $(M_s+1)$-th sub-slots, the BS will individually push the files, $f_i$, $i\in\{1, \cdots, M_s\}$,  to the content servers. Compared to this more sophisticated OMA scheme, the use of the proposed push-and-deliver strategy still offers a significant gain in terms of the cache hit probability, as will be shown in Section \ref{section simulation}.  

\subsection{Extension to D2D Caching}

The aim of this subsection is to show that the concept of push-and-deliver can also be applied to D2D caching. Assume that a time slot is dedicated to a user whose request cannot be found in the caches of its neighbours, and during this time slot, the BS will send the requested  file $ {f}_0$ to the user directly. By applying the push-and-deliver strategy, the BS will also proactively push $M_s$ new files, $f_l$, $1\leq l\leq M_s$, to all  users for future use.  In other words, when the BS addresses the current demand of a user directly, the BS pushes   more content files  to all users, including the user which requests $f_0$, for future use.   

In the context of D2D caching, content servers are no longer needed. Therefore, the spatial model presented in Section \ref{Section system model} needs to be revised accordingly. Particularly, it is assumed  that  the locations of the users  are denoted by $y_k$ and modelled as an HPPP, denoted by $\Phi_u$, with density $\lambda_u$.  

After implementing   the push-and-deliver strategy, following  similar  steps as   in the previous subsection,   for a user with   distance $r$ from the BS and   Rayleigh fading channel gain $h$, the outage probability for decoding  $f_i$ is given by
\begin{align}\label{rate1}
\text{R}^i_r  =& \log\left(1+ \frac{   \alpha_i^2|h| ^2 r^{-\alpha}  }{  \sum^{M_s}_{j=i+1}  \alpha_j^2|h |^2r^{-\alpha}   +\frac{1}{\rho}}  \right),
\end{align}
when $\text{R}^l_r>R_l$, for $0\leq l\leq i-1$. If $f_{j}$, $0\leq j\leq M_s-1$, can be decoded correctly,   $f_{M_s}$ can be decoded by this user with the following data rate: 
\begin{align}\label{rate1}
\text{R}^{M_s}_r  =& \log\left(1+   \rho \alpha_{M_s}^2 |h |^2r^{-\alpha}    \right).
\end{align}
Consequently, for a user with   distance $r$   from the BS, the    probability of successfully decoding  $f_i$ can be expressed as follows:
\begin{align}\label{thinning}
\mathrm{P} ^i(r) =&   \mathrm{P}\left(\text{R} ^l_r>R_l, \forall l\in\{0, \cdots, i\}\right)\\\nonumber 
=&   e^{-\bar{\tau}_i^\alpha r^\alpha }. 
\end{align}
 
Following \eqref{thinning},  one can draw  the conclusion that        the locations of the users  which can successfully receive $f_i$ no longer follow the original HPPP with $\lambda_u$, but follow an inhomogeneous PPP which is thinned from the original HPPP by $\mathrm{P} ^i(r) $, i.e., the density of this new PPP is $\mathrm{P} ^i(r) \lambda_u$. By using this thinning process, the cache miss probability can be characterized as follows. 

During the  D2D content delivery phase, assume    a newly  arrived user, whose location is denoted by $y_0$, requests   file $f_i$. Denote  by $\mathcal{B}(y_0, d)$   a disc with radius $d$ and   its origin located at $y_0$.  This disc is the  area in which the   user   searches for a helpful neighbour which has the  requested file in its cache. For this inhomogeneous PPP, the cache hit probability for the user requesting $f_i$ is given by
\begin{align}\label{hit missing}
\mathrm{P}^{hit}_i =& 1 - \mathrm{P}(\text{no   user in } \mathcal{B}(y_0, d) \text{ caches } f_i)
\\\nonumber =& 1 -  e^{-\Lambda_i (\mathcal{B}(y_0, d))},
\end{align}
where   $\Lambda_i (\mathcal{B}(y_0, d))$ denotes the intensity measure of the inhomogeneous PPP for the users which have $f_i$ in their caches.   In \eqref{hit missing}, the hit probability is found by determining   the cache miss probability which corresponds to the event that  the user cannot find its requested file in the caches of its neighbours located in the disc $\mathcal{B}(y_0, d)$.  The calculation of the cache  hit probability depends  on the relationship between $d$ and the distance between  the observing user and the BS, denoted by $r_0$, as shown in the following subsections.

\begin{figure}[t]\vspace{-1em}
\begin{center} \subfigure[ $d<r_0$ ]{\label{fig 0 a}\includegraphics[width=0.44\textwidth]{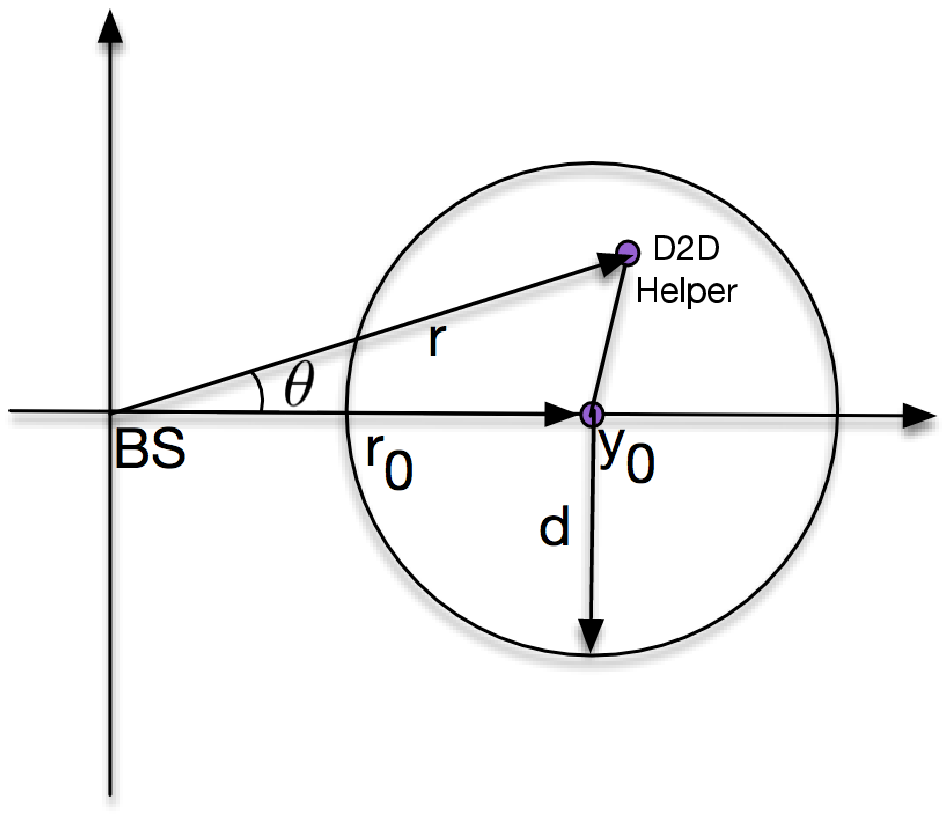}}
\subfigure[$d\geq r_0$ ]{\label{fig 0 b}\includegraphics[width=0.44\textwidth]{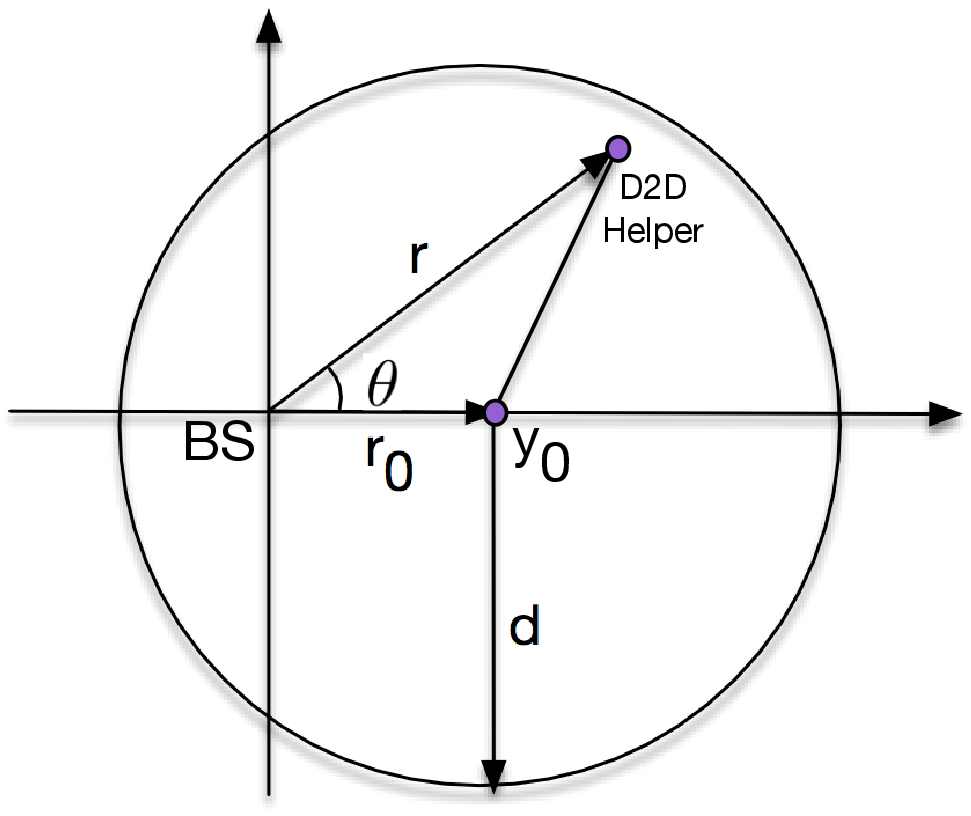}}\vspace{-1em}
\end{center}
\caption{ Two possible cases between  the radius of the search  disc, $\mathcal{B}(y_0, d)$,    and the distance between the observing user located at $y_0$ and the BS.  }\label{fig0}\vspace{-0.5em}
\end{figure}
\subsubsection{For the case of $d< r_0$} For $d<r_0$, define $\Lambda_i (\mathcal{B}(y_0, d))\triangleq \Lambda^i_{d\leq r_0}(r_0) $.
The assumption $d<r_0$ means that the BS is excluded  from  $\mathcal{B}(y_0, d)$.  Therefore, the intensity measure can be calculated as follows:
\begin{align}
 \Lambda^i_{d\leq r_0}(r_0)  
 &= \underset{{r, \theta \in \mathcal{B}(y_0, d)}}{\int\int }\mathrm{P}^i (r) \lambda_u d\theta rdr.
\end{align}
As can be observed from Fig. \ref{fig0}, the constraint on $r$ and $\theta$ can be  expressed as follows:
\begin{align}\label{theta app}
r^2+r_0^2 -2r_0 r\cos \theta\leq d^2.
\end{align}
Therefore, the intensity measure can be expressed  as follows:
\begin{align}\nonumber 
\Lambda^i_{d\leq r_0} (r_0)
 &=\int^{r_0+d}_{r_0-d} \int^{\arccos \frac{r^2+r_0^2-d^2}{2r_0r}}_{-\arccos \frac{r^2+r_0^2-d^2}{2r_0r}}\mathrm{P}^i (r) \lambda_u d\theta rdr\\ 
  &=2\lambda_u\int^{r_0+d}_{r_0-d}  \mathrm{P} ^i(r)    r\arccos \frac{r^2+r_0^2-d^2}{2r_0r} dr. 
\end{align}
By applying Chebyshev-Gauss quadrature, the intensity measure can be calculated as follows:
\begin{align}  \label{case 1}
\Lambda^i_{d\leq r_0} (r_0)
 &\approx 2\lambda_u  d \sum^{N}_{l=1}\frac{\pi}{N} g_r\left(r_0+dw_l\right) \sqrt{1-w_l^2},
\end{align}
where   $g_r(z)$ is given by
\begin{align}
g_r(z) = \mathrm{P}^i (z)    z\arccos \frac{z^2+r_0^2-d^2}{2r_0z} .
\end{align}
\subsubsection{For the case of $d\geq r_0$}
 For $d\geq r_0$,  define $\Lambda ^i(\mathcal{B}(y_0, d))\triangleq \Lambda_{d> r_0} $. The assumption,  $d\geq r_0$, means that the BS is   inside of $\mathcal{B}(y_0, d)$.   Following  similar  steps  as in  the previous case, the intensity measure can be evaluated  as follows:
 \begin{align}\nonumber 
\Lambda^i_{d\leq r_0} (r_0)
 =&\int_{0}^{d-r_0} \int^{\pi}_{-\pi}\mathrm{P}_i (r) \lambda_u d\theta rdr\\\nonumber &+\int^{r_0+d}_{d-r_0} \int^{\arccos \frac{r^2+r_0^2-d^2}{2r_0r}}_{-\arccos \frac{r^2+r_0^2-d^2}{2r_0r}}\mathrm{P}_i (r) \lambda_u d\theta rdr\\ \nonumber
  \approx&\frac{2\pi \lambda_u }{\alpha\bar{\tau}_i^2} \gamma\left( \frac{2}{\alpha}, \bar{\tau}_i^\alpha(d-r_0)^\alpha \right) \\&+2\lambda_u  r_0 \sum^{N}_{l=1}\frac{\pi}{N} g_r\left(d+r_0w_l\right) \sqrt{1-w_l^2},  \label{case 2}
\end{align}
where $\gamma(\cdot)$ denotes the incomplete Gamma function, and the approximation in the last step follows from the application of Chebyshev-Gauss quadrature.

Finally,  the cache hit probability can be obtained by substituting \eqref{case 1} and \eqref{case 2} into \eqref{hit missing}.

\section{Numerical Studies and Discussions } \label{section simulation}
In this section, the performances achieved by the proposed push-then-deliver and push-and-deliver strategies are studied by using computer simulations, where the accuracy of the developed analytical results will be also verified. The system parameters adopted for simulation and analysis are specified in the captions of the figures shown in this section. 

\subsection{Performance of Push-then-deliver Strategy}
In Figs. \ref{fig1} and \ref{fig2}, the impact of the NOMA assisted  push-then-deliver strategy on the cache hit probability is studied. The thermal noise is set as  $\sigma_n^2 = -100$ dBm.  The density of the content servers is parameterized by cluster radius  $\mathcal{R}_c$, i.e., $\lambda_c= \frac{0.01}{\pi \mathcal{R}_c^2}$, in order to account for  the fact that the density of the content servers is affected by $\mathcal{R}_c$.  By applying the NOMA principle to the content pushing phase, more content can be pushed to the content servers simultaneously, and hence, the cache hit probability is improved, compared to the OMA case, as can be observed from Fig. \ref{fig1}. For example, when the transmission power is $40$ dBm,    the shape parameter for  the content popularity is $\gamma=0.5$, and $\mathcal{R}_c=50$ m, the use of OMA yields  a hit probability of $0.2$, and the use of NOMA improves this value to  $0.45$, which corresponds to a $100\%$ improvement. At low SNR,  NOMA and OMA yield the same performance. This is due to the use of the CR  inspired power allocation policy in \eqref{c1}, which implies    that at low SNR, all the power is allocated to file $f_1$, and hence, there is no difference between the OMA and NOMA schemes.   Note that the curves for analysis and simulation match perfectly in  Fig. \ref{fig1},  which demonstrates the accuracy of the developed analytical results.   Furthermore, we note that   the  NOMA power allocation coefficients $\beta_i$ are assumed to be fixed. Optimizing these power allocation coefficients dynamically according to the users' channel conditions can further enhance  the performance gain of NOMA assisted caching compared to  the OMA baseline  scheme.

\begin{figure}[!htp]\vspace{-1em}
\begin{center} \subfigure[ $\mathcal{R}_c=100$m ]{\label{fig 1 a}\includegraphics[width=0.64\textwidth]{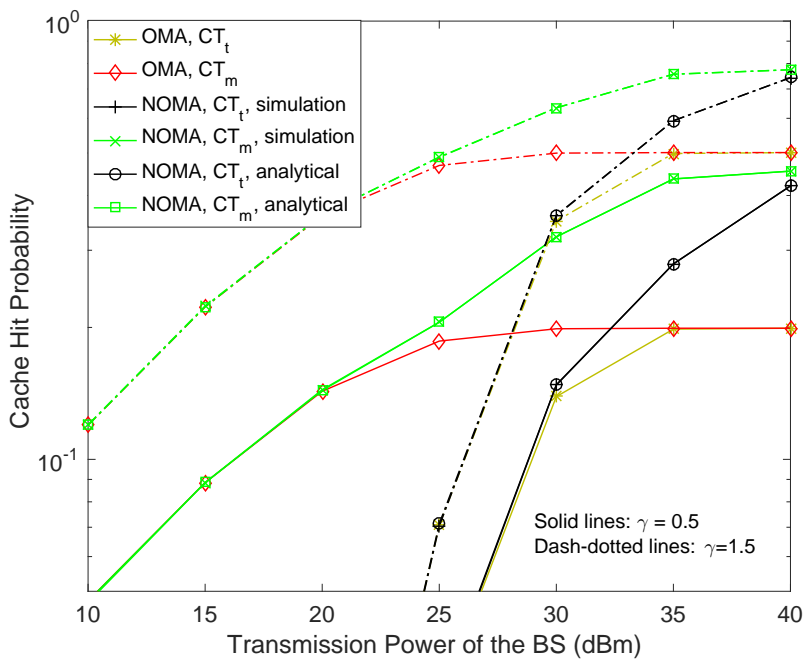}}
\subfigure[$\mathcal{R}_c=50$m ]{\label{fig 1 b}\includegraphics[width=0.64\textwidth]{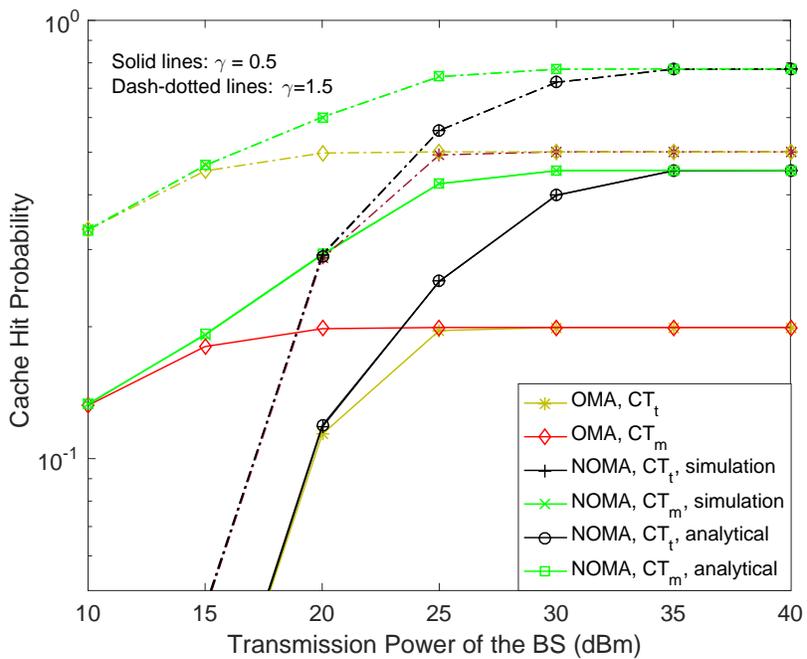}}\vspace{-1em}
\end{center}
\caption{  The cache hit probability  for     the push-then-deliver strategy.   $N=20$, $\alpha=3$,  $\lambda_c= \frac{0.01}{\pi \mathcal{R}_c^2}$,    $t=5$,  $m=1$,   $M_s=3$,  and $R_l=1$ bit per channel use (BPCU),  for $1\leq l \leq 3$. The power allocation coefficient for file $f_1$ is based on the CR power allocation policy.   The power allocation coefficients for files $f_2$ and $f_3$ are $\beta_2=\frac{3}{4}$ and $\beta_3=\frac{1}{4}$, respectively.      $|\mathcal{F}|=10$. }\label{fig1}\vspace{-0.5em}
\end{figure}

The impact of $\gamma$, the shape parameter defining  the content popularity, on the hit probability is  significant, as can be observed in Fig. \ref{fig1}. Particularly, increasing the value of $\gamma$   improves the hit probability. This is because a larger value of $\gamma$ implies that the first $M_s$ files become more popular,   hence ensuring the delivery of these more popular files   can significantly improve the hit probability, as indicated by  \eqref{hit pr}. Comparing Fig. \ref{fig 1 a} with Fig. \ref{fig 1 b}, one can observe that  the impact of $\mathcal{R}_c$ on the hit probability is also significant, which is due to the fact that the density of the content servers depends on   $\mathcal{R}_c$. Particularly,  a larger $\mathcal{R}_c$ means that the content servers  are more sparsely deployed and hence it is more difficult for the BS to push content to the  content  servers, and the cache hit probability decreases. 

 Recall that during the  time slot considered for content pushing  in Section \ref{subsection IIIA}, the BS pushes additional files to  $\text{CS}_m$ while ensuring that $f_1$ is pushed to $\text{CS}_t$.   In Fig. \ref{fig2}, the impact of different  choices of $m$ and $t$ on the cache hit probability is studied. As can be observed from the figure, increasing $t$ will decrease the hit probability. This is again due to the use of the CR power allocation policy.  In particular,  a larger $t$ means that more transmission power is needed to delivery $f_1$ to $\text{CS}_t$, and hence, less power is available for other files. An interesting observation in Fig. \ref{fig2} is that the shape of the hit probability curves is not smooth. This is due to the fact that the hit probability is the summation of   popularity probabilities   $P(f_l)$ and these  popularity probabilities are  prefixed and not continuous, as shown in \eqref{popular}.  On the other hand, for a fixed  $t$, increasing $m$ reduces the cache hit probability, since increasing $m$ means that $\text{CS}_m$ is further away from the BS and hence its reception reliability deteriorates.   

\begin{figure}[!tp]\centering\vspace{-1em}
    \epsfig{file=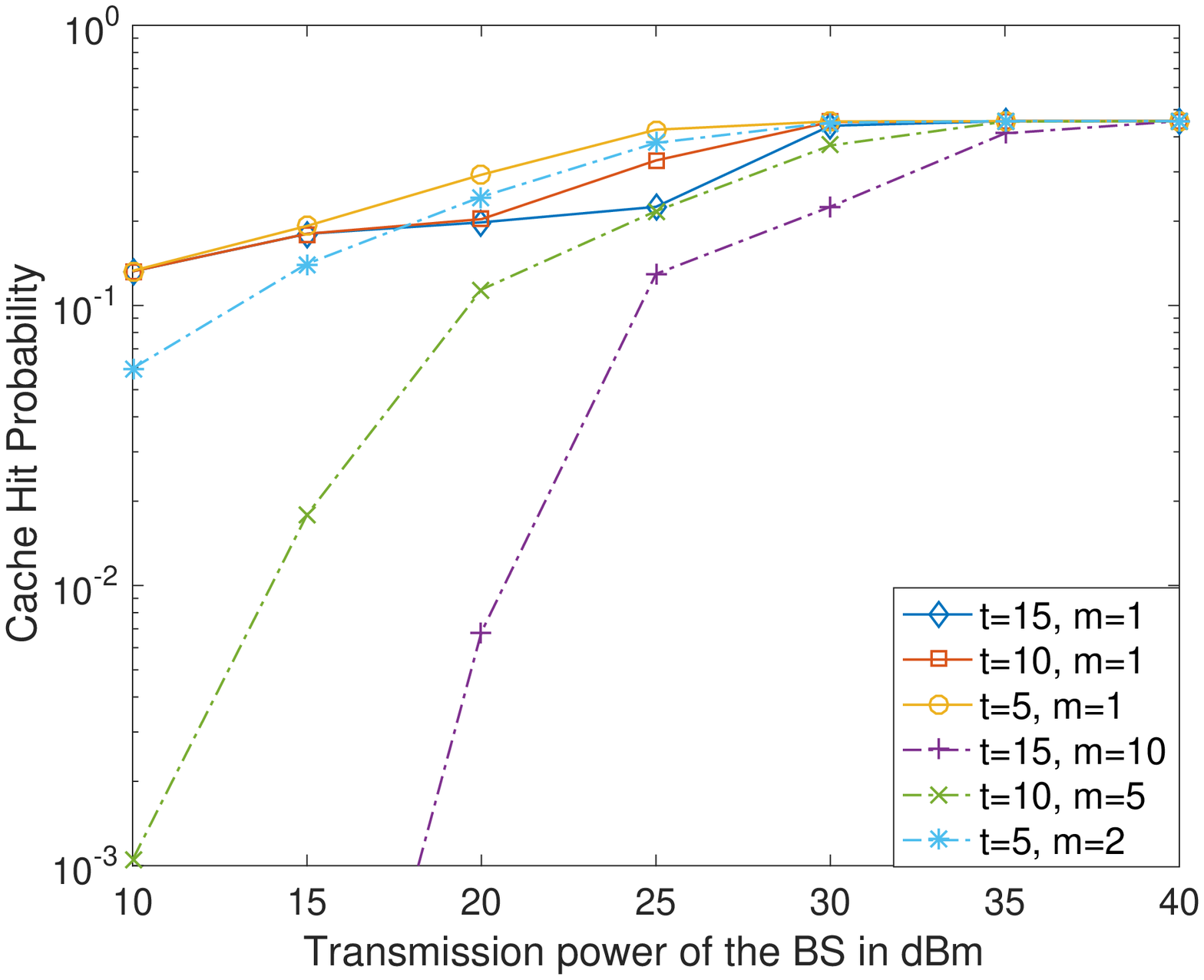, width=0.64\textwidth, clip=}\vspace{-1em}
\caption{ The impact of the choices of $m$ and $t$ on the cache hit probabilities for   the push-then-deliver strategy.  $N=20$, $\alpha=3$,  $\mathcal{R}_c=50$m,  $\lambda_c= \frac{0.01}{\pi \mathcal{R}_c^2}$,  ${M_s}=3$,  $t=5$,   $m=1$, and $R_l=1$ BPCU,  for $1\leq l \leq 3$. The power allocation coefficient for file $f_1$ is based on the CR power allocation policy.   The power allocation coefficients for files $f_2$ and $f_3$ are $\beta_2=\frac{3}{4}$ and $\beta_3=\frac{1}{4}$, respectively.      $\gamma=0.5$ and   $|\mathcal{F}|=3$. Analytical results are used to generate the figure. \vspace{-0.5em} }\label{fig2}
\end{figure}

\begin{figure}[!htp]\vspace{-1em}
\begin{center} \subfigure[ $\alpha=3$ ]{\label{fig 3 b1}\includegraphics[width=0.64\textwidth]{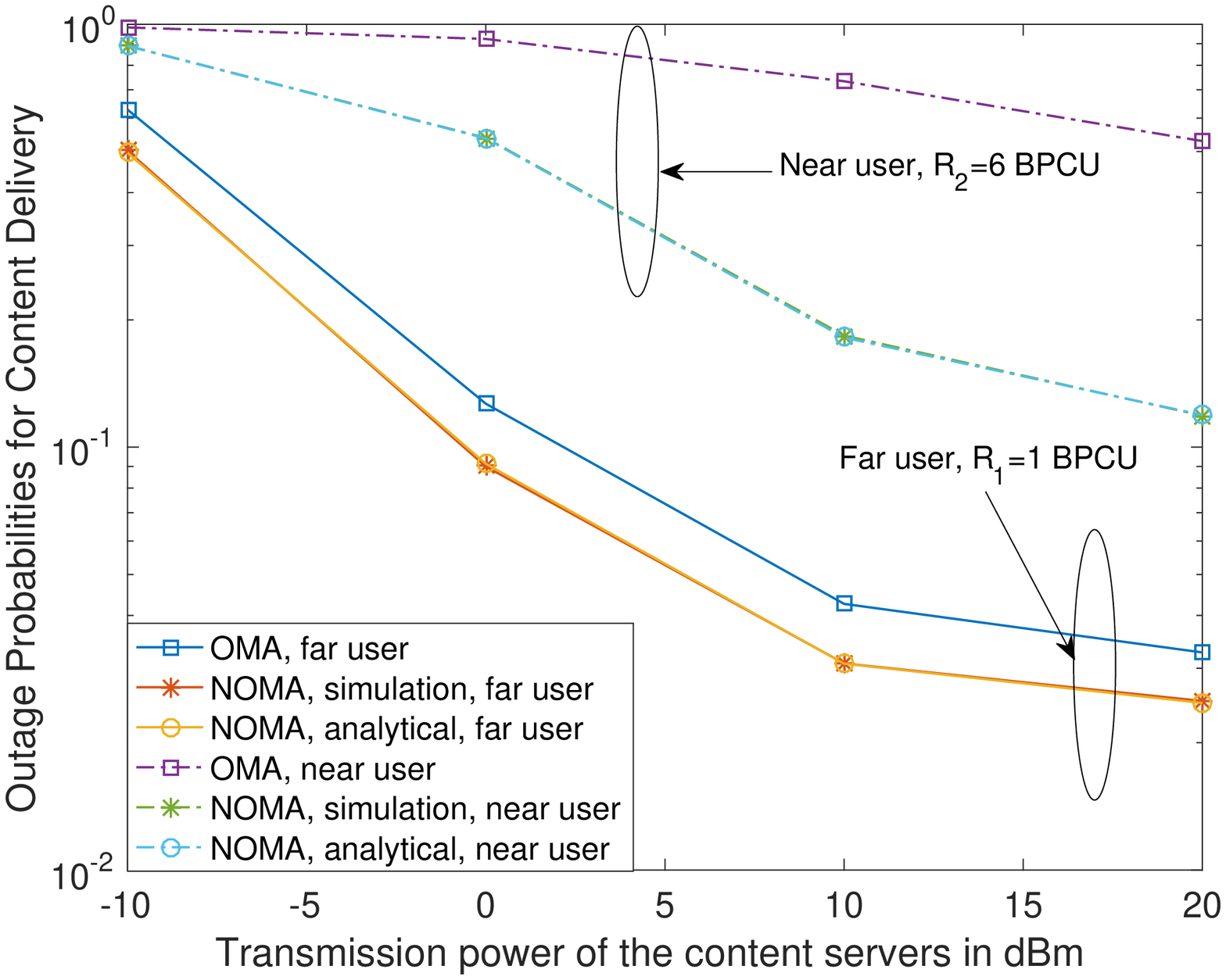}}
\subfigure[$\alpha=4$ ]{\label{fig 3 b2}\includegraphics[width=0.64\textwidth]{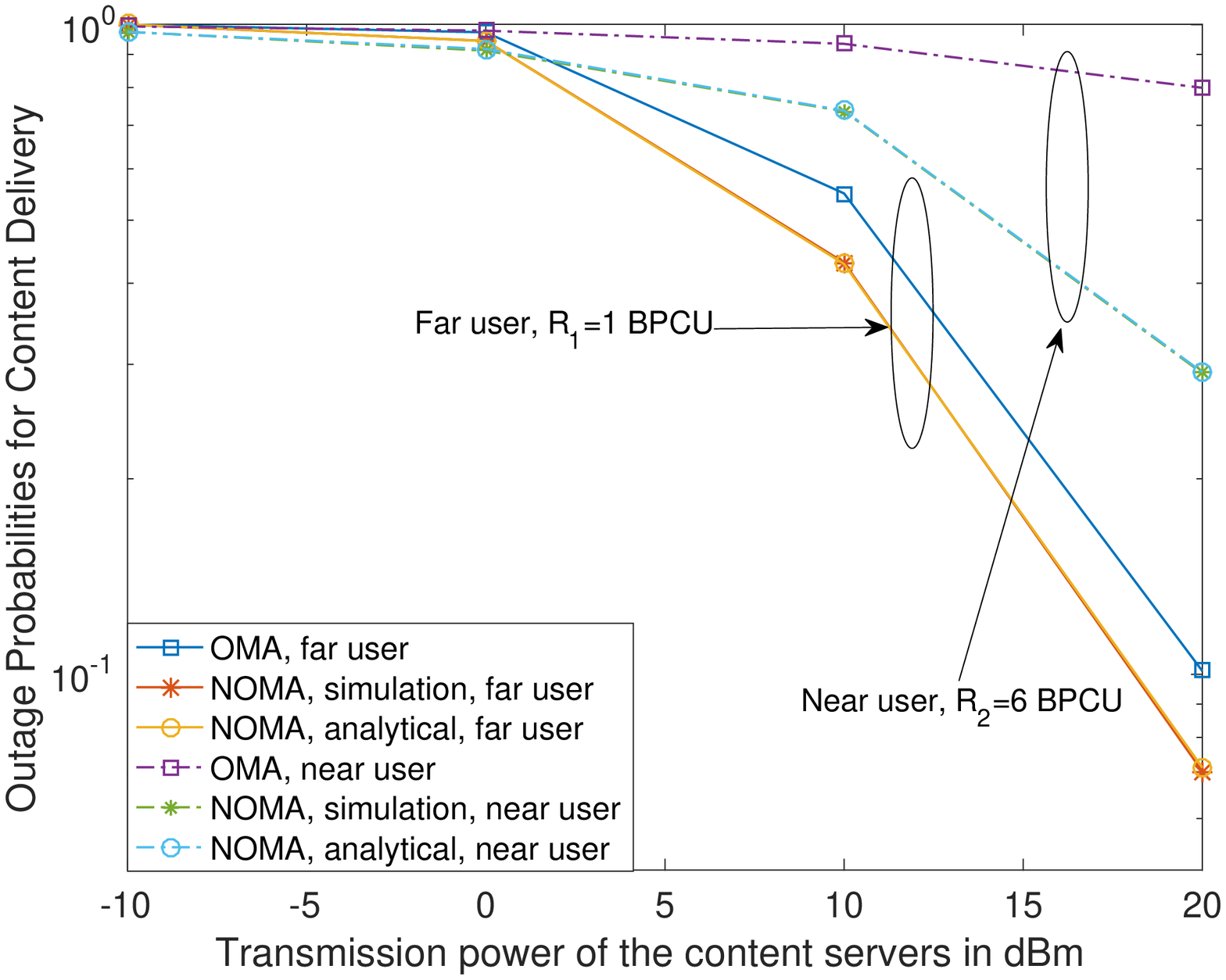}}\vspace{-1em}
\end{center}
\caption{  The outage probabilities for   content delivery for the push-then-deliver strategy.  $N=20$, $\alpha=4$, $\mathcal{R}_c=100$m.  $\lambda_c= \frac{0.01}{\pi \mathcal{R}_c^2}$,     $R_1=1$ BPCU, and $R_2=6$ BPCU. The power allocation coefficients are $\alpha_1^2=\frac{3}{4}$ and $\alpha_2^2=\frac{1}{4}$.   }\label{fig3}\vspace{-0.5em}
\end{figure}
In Fig. \ref{fig3}, the impact of using NOMA   for content delivery is studied,  where the rate pair $\{R_1, R_2\}$ is set to $\{1,6\}$ BPCU to account for  the fact that the near user can achieve a higher  data rate. As can be observed from the figure, the proposed push-then-deliver strategy can improve the  reliability of content delivery, particularly for the user with strong channel conditions. For example, when the path loss exponent is set to $\alpha=3$ and the transmission power of the content servers is $20$ dBm, the use of NOMA ensures that the outage probability for  the far user  is improved from $4.5\times 10^{-2}$ to $3\times 10^{-2}$, which is a relatively small performance gain. However,  the performance gap between the OMA and NOMA schemes for the near user is much larger, e.g., for the same case as considered  before, the outage probability is improved from $5\times 10^{-1}$ to $1.1\times 10^{-2}$. Note that the outage probability for content delivery has an error floor, i.e., increasing the transmission power of the content severs cannot reduce the outage probability to zero. This is because   multiple content servers transmit simultaneously, and hence,  content delivery becomes interference limited at high SNR. We note  that the impact of the path loss exponent on the reliability of content delivery is significant, as can be observed by comparing Figs. \ref{fig 3 b1} and \ref{fig 3 b2}. This is due to the fact that a smaller value of $\alpha$ results in a lower  path loss, which leads to an improved  reception reliability.

\begin{figure}[!htp]\vspace{-1em}
\begin{center} \subfigure[ Case 1]{\label{fig 4
b1}\includegraphics[width=0.64\textwidth]{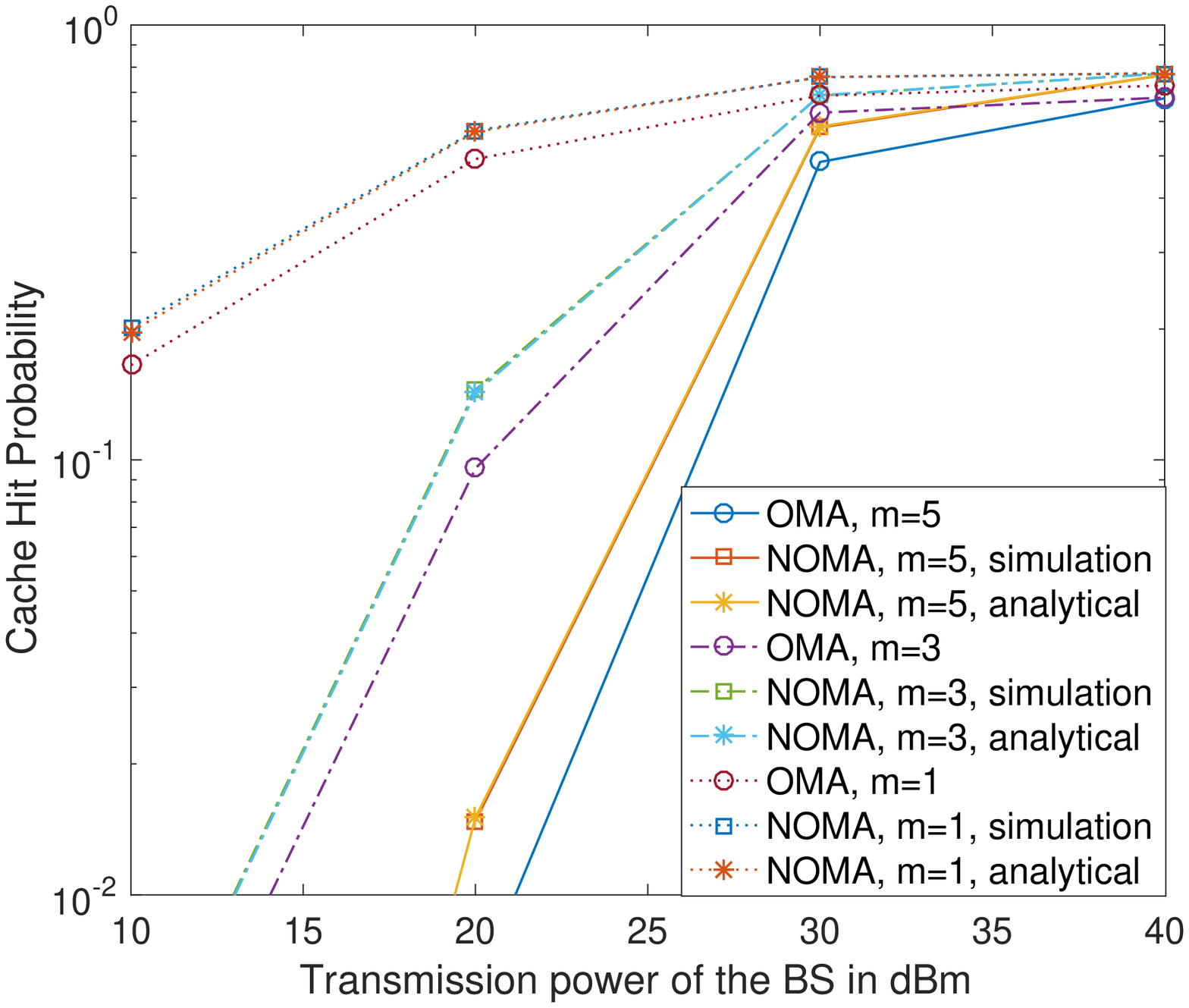}}
\subfigure[Case 2 ]{\label{fig 4
b2}\includegraphics[width=0.64\textwidth]{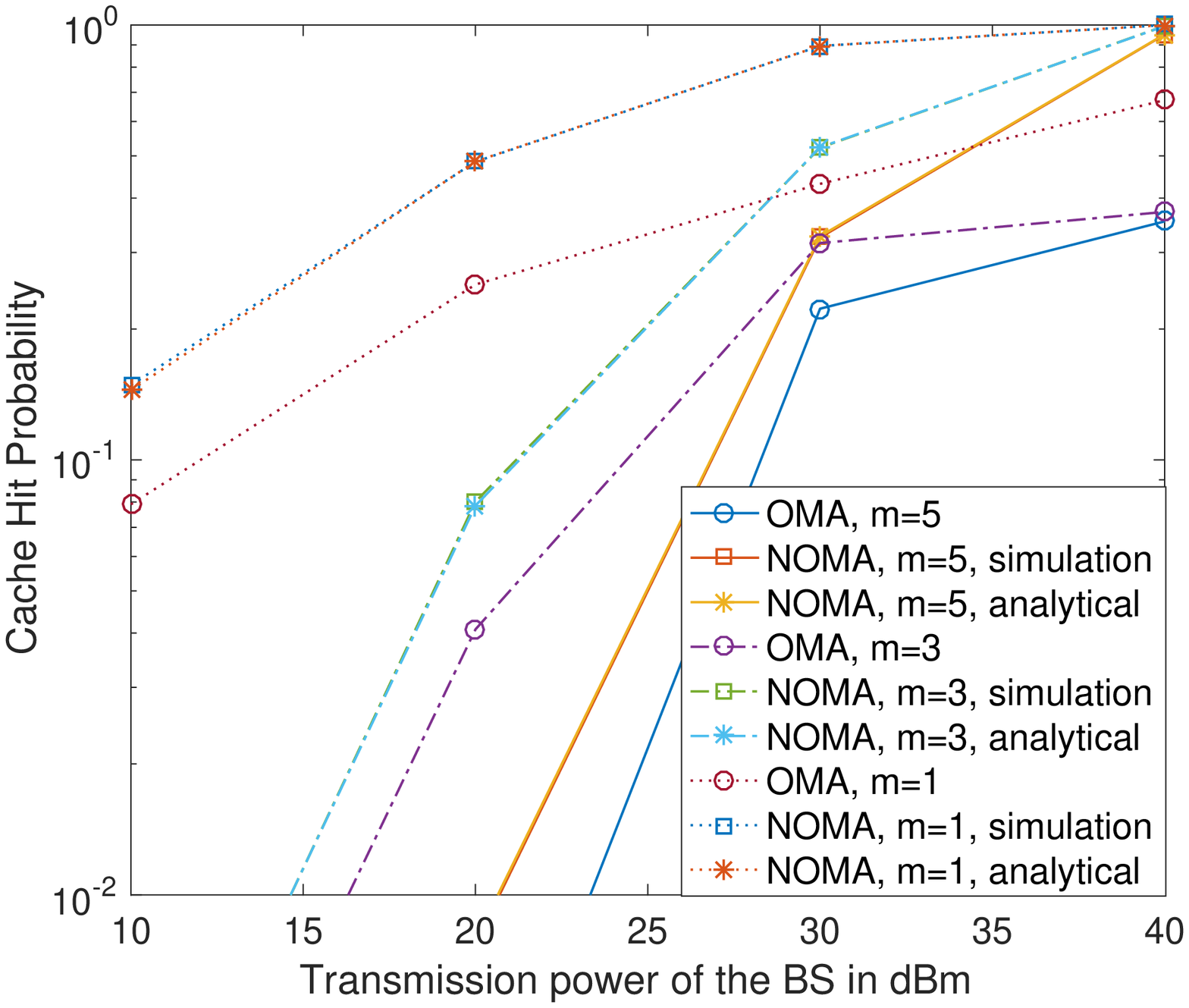}}\vspace{-0.5em}
\end{center}
\caption{  The cache hit probability   for the proposed     push-and-deliver strategy. $\mathcal{R}_c=50$m, $N=20$, $\gamma=1.5$, $\alpha=3$, $\lambda_c= \frac{0.01}{\pi \mathcal{R}_c^2}$, $M_s=3$, $\delta=1.1$.    $R_0=\frac{1}{8}$ BPCU, $R_1=\frac{3}{4}$ BPCU, $R_2=\frac{7}{8}$ BPCU,  and $R_3=\frac{11}{4}$ BPCU.   The power allocation coefficients are $\alpha_0^2=\frac{4}{8}$, $\alpha_1^2=\frac{3}{8}$, $\alpha_2^2=\frac{2}{8}$, and $\alpha_3^2=\frac{1}{8}$.  }\label{fig4} 
\end{figure}

\subsection{Performance of Push-and-deliver Strategy}
In Fig. \ref{fig4}, the impact of the proposed push-and-deliver strategy on the cache hit probability is studied.  We employ   $\delta=1.1$   to avoid the trivial case where the user is located at the same place as the BS, as discussed in  Section \ref{subs detl}.  As can be observed, the use of the proposed strategy can effectively  improve the cache hit probability compared to the OMA case, which is consistent with the conclusions drawn  in the previous subsection. In both sub-figures of Fig. \ref{fig4}, the analytical results  perfectly match    the simulation results, which verifies  the accuracy of the developed analysis. 

 In Fig. \ref{fig4},  the impact of different choices for the popularity parameters on the cache hit probability is   studied. In particular, the following two cases are considered:
 \begin{itemize} 
 \item Case 1: $ {\mathcal{F}}_1=\{f_1,\cdots ,f_{10}\}$, and the power allocation coefficient for $f_l$ is $\alpha_l$;
 \item Case 2: $ {\mathcal{F}}_2=\{f_1,\cdots,    f_3\}$, and the power allocation coefficient for $f_l$ is $\alpha_{4-l}$.
 \end{itemize}
 The two cases correspond to two different options for mapping  files with different popularities  to different power levels (or equivalently SIC decoding orders), where in the first case,  more popular files are assigned more power,  and in the second case,  less power is assigned to more popular files.

 In  Case 1,   the performance gap between NOMA and OMA is not significant, as can be observed from Fig. \ref{fig 4
b1}.  For example, when the transmit power is $40$ dBm and $m=5$, the use of OMA results in a hit probability of $0.7$ and the use of NOMA yields  a hit probability of $0.8$, where the gap is $0.1$ only. However, for a different set of   popularity parameters, i.e., Case 2, the performance gap between OMA and NOMA is significantly increased. For example, for a  transmit power of $40$ dBm and $m=5$, the performance gap between OMA and NOMA is enlarged to $0.5$. The reason behind  this phenomenon is  as follows. Recall that the use of NOMA can significantly improve the reception reliability of the files which are decoded at the later stages of  the SIC procedure, but the improvement for the files which are decoded during the first few stages  of SIC is not significant.  In Case 1,  the first few files will get larger weights in the sum of the cache hit probability, i.e., file $f_l$, for a small $l$, has more impact on the overall performance. As a result, the gap between OMA and NOMA in Case 1 is small, since the reception reliability for decoding these   files in the case of NOMA is not so different from that for OMA. On the other hand, Case 2 means that the most popular file, $f_1$, will be decoded   last. As discussed before,  the capabilities of OMA and NOMA to decode $f_1$ are quite different, which is the reason for the larger performance gap in   Case 2. 

Recall that the key idea of the push-and-deliver strategy is to perform content pushing when asking the BS to serve the users directly.  Fig. \ref{fig4} clearly demonstrates that this strategy can   efficiently push new content to the content servers, but it does not demonstrate  the impact of this strategy on content delivery, which is studied in Fig. \ref{fig5}. Particularly, as can be observed from  the figure, the use of the proposed push-and-deliver strategy does not degrade the reception reliability of content delivery. In fact, the use of NOMA can even  improve the outage probability for content delivery. 

In Fig. \ref{fig6}, the concept of the proposed push-and-deliver strategy is extended to D2D caching scenarios. Without loss of generality, the newly arrived user  is located at $y_0=(500 \rm{m}, 500 m)$. As expected, the use of the proposed strategy can significantly reduce the  miss probability, compared to the case of OMA. For example, for the case where the user density is $\lambda_u=5\times 10^{-5}$, a transmit power of $40$ dBm, and $d=150$ m, the use of NOMA yields a miss probability of $6\times 10^{-2}$, whereas  the miss probability for OMA is $1.6\times 10^{-1}$, which is much worse. As can be observed from the figure, increasing the value of $d$ can reduce the miss probability, since the area for searching for a D2D helper is increased. Another important observation is that by increasing the density of the users, the miss probability can be further reduced, since increasing the density means that more users are located in the same area and hence it is more likely to find a D2D helper. We note that, in Fig. \ref{fig6}, computer simulation and analytical results match perfectly, which demonstrates again the accuracy of the developed analysis. 

\begin{figure}[!htbp]\centering\vspace{-1em}
    \epsfig{file=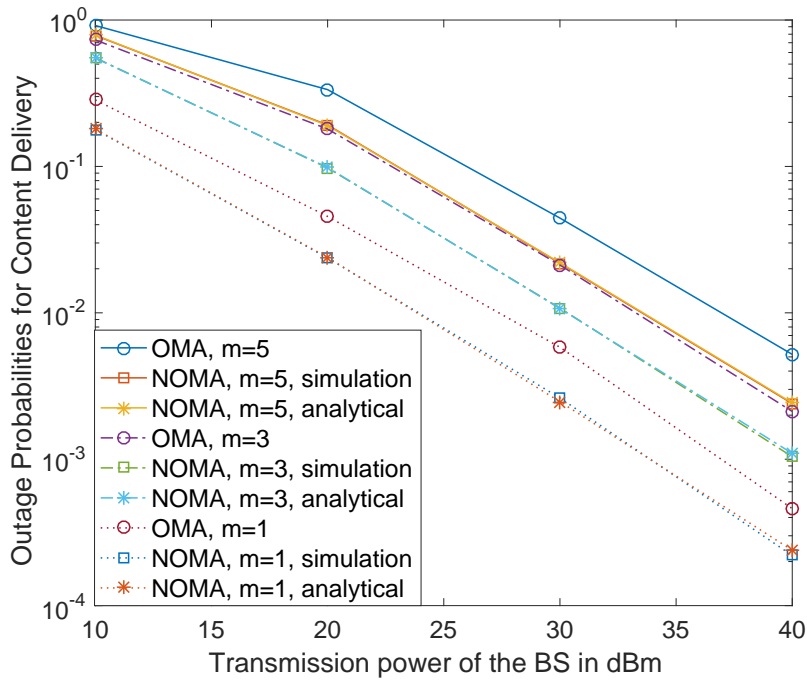, width=0.64\textwidth, clip=}\vspace{-1em}
\caption{ The impact of the push-and-deliver strategy on content delivery. $\mathcal{R}_c=50$m, $\gamma=1.5$, $\alpha=3$, $\lambda_c= \frac{0.01}{\pi \mathcal{R}_c^2}$, $M_s=3$, $N=20$, $\delta=1.1$.   $R_0=\frac{1}{8}$ BPCU, $R_1=\frac{3}{4}$ BPCU, $R_2=\frac{7}{8}$ BPCU,  and $R_3=\frac{11}{4}$ BPCU.   The power allocation coefficients are $\alpha_0^2=\frac{4}{8}$, $\alpha_1^2=\frac{3}{8}$, $\alpha_2^2=\frac{2}{8}$, and $\alpha_3^2=\frac{1}{8}$. Case 2 is considered.   \vspace{-1em} }\label{fig5}
\end{figure}

 \begin{figure}[!htp]\vspace{-1em}
\begin{center} \subfigure[ $\lambda_u=5\times10^{-5}$ ]{\label{fig 6 b1}\includegraphics[width=0.64\textwidth]{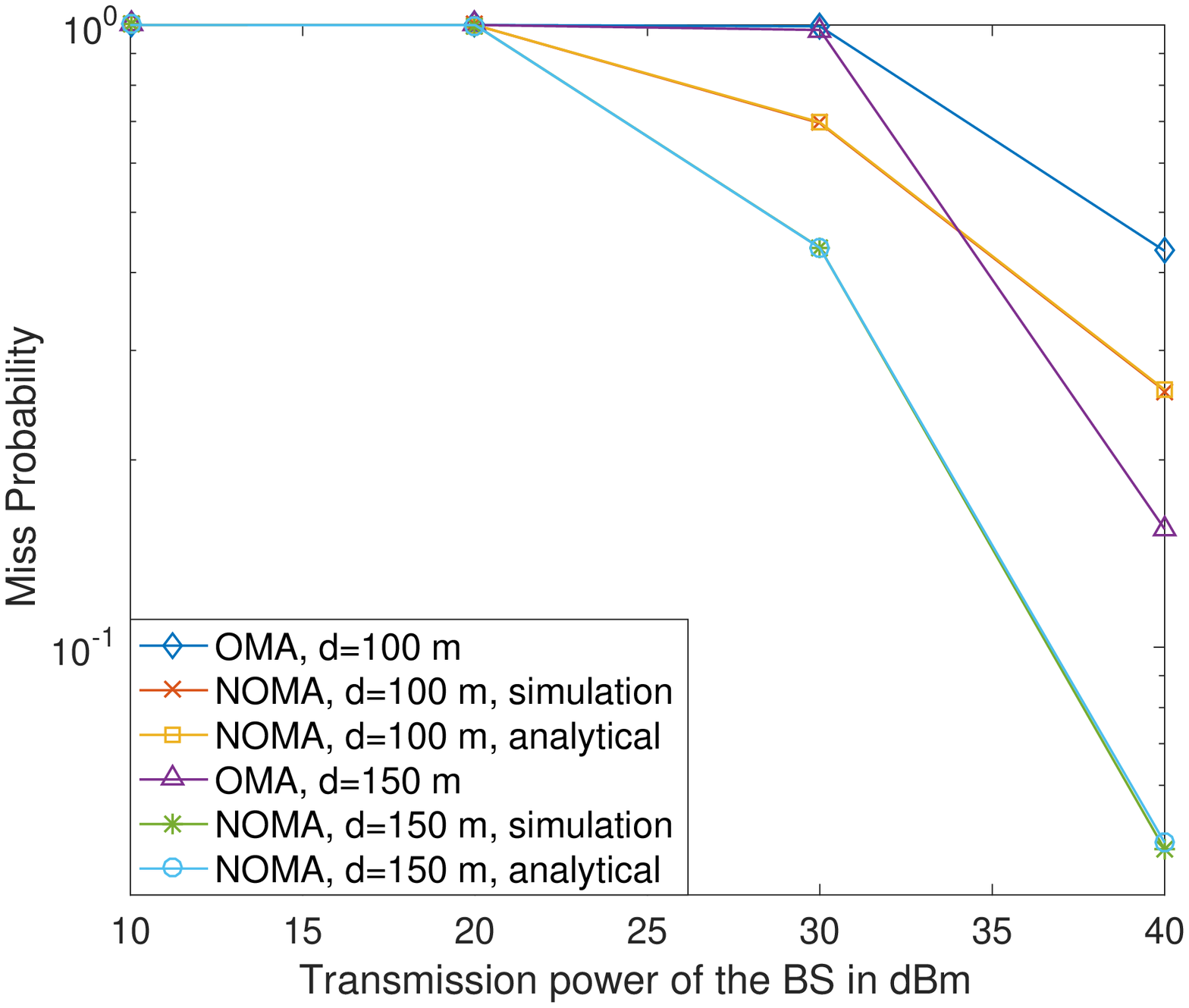}}
\subfigure[$\lambda_u=1\times10^{-4}$ ]{\label{fig  6 b2}\includegraphics[width=0.64\textwidth]{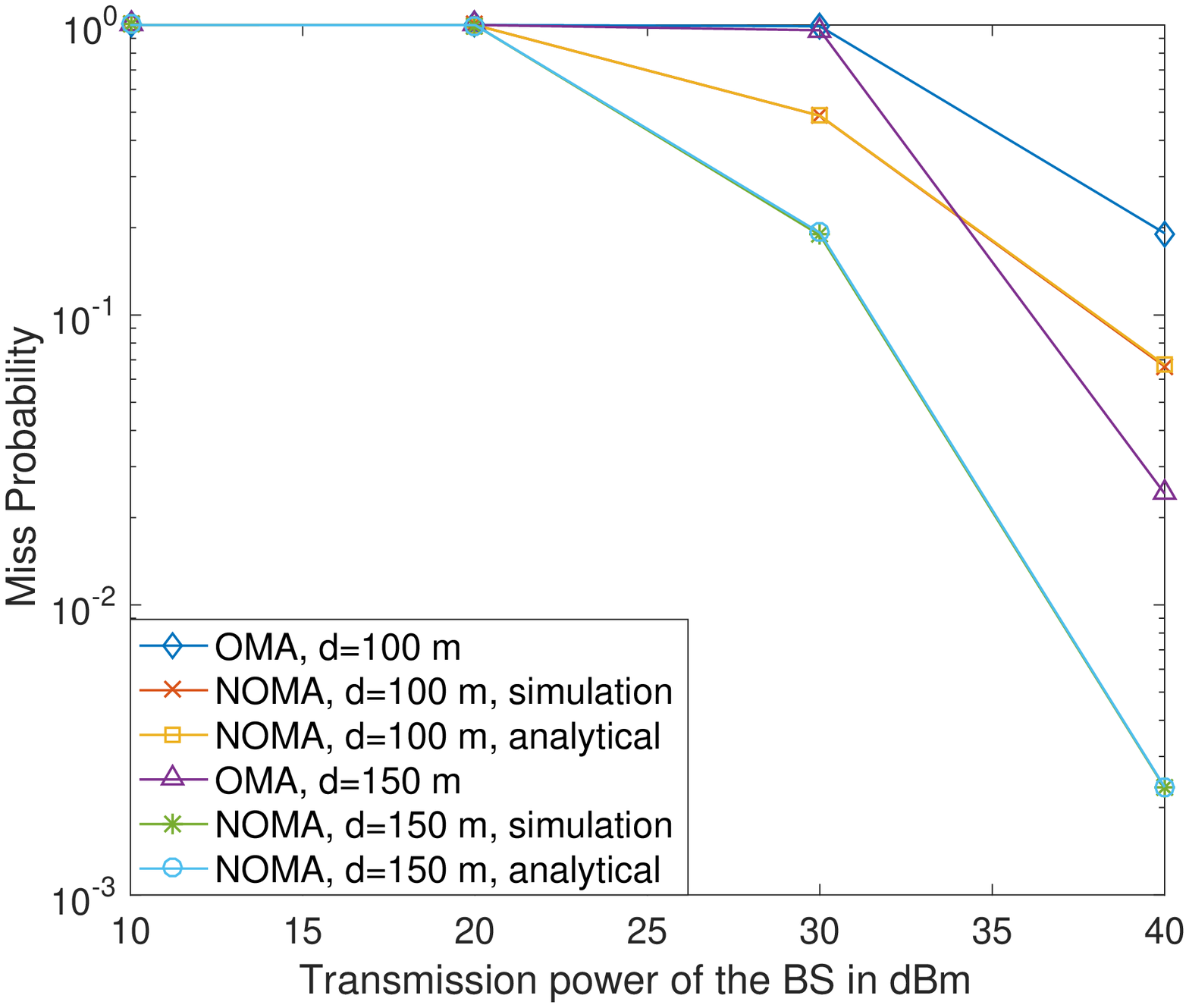}}\vspace{-1em}
\end{center}
\caption{  The impact of the proposed  push-and-deliver strategy on the cache miss probability in D2D caching scenarios.  $N=20$,  $\alpha=3$,   $y_0=$ {\rm (500 m, 500 m)},  $R_0=0.5$ BPCU, $R_1=4$ BPCU and $M_s=1$. The power allocation coefficients are $\alpha_0^2=\frac{3}{4}$ and $\alpha_1^2=\frac{1}{4}$.   }\label{fig6}\vspace{-1em}
\end{figure}

\section{Conclusions} \label{section conclusions}
 Unlike conventional wireless caching strategies which rely  on the use of off-peak hours for content pushing, in this paper,  the NOMA principle has been applied to wireless caching to enable the frequent update of the local caches  via wireless transmission during on-peak hours. Two NOMA assisted caching strategies have been developed, namely the push-then-deliver strategy and the push-and-deliver strategy.   The push-then-deliver strategy is applicable to the case when  the content pushing phase and the content delivery phase are separated, and utilizes  the NOMA principle independently in both  phases. The developed analytical results demonstrate that the proposed NOMA assisted caching scheme can efficiently improve the cache hit probability and reduce the delivery outage probability. The push-and-deliver strategy is motivated by the fact that, in practice,  it is inevitable that some user requests cannot be accommodated locally and   the BS has to serve these  users directly. The key idea of the push-and-deliver strategy is to merge the two phases, i.e., the BS  pushes content to the content servers while simultaneously serving users directly. Furthermore, in addition to the caching scenario with caching infrastructure, e.g., content servers, we have    considered   D2D caching, where  the use of NOMA has also been shown to yield   superior performance compared to   OMA. 

The results in this paper open  several new   directions for future research. First, in this paper, the file popularity parameter has been assumed to be given and fixed. As demonstrated in Fig. \ref{fig4}, different choices for this parameter yield different cache hit probabilities, which means that dynamically optimizing the   NOMA power allocation (or equivalently the NOMA SIC decoding order) for given  content  popularity parameters   could  further improve the performance of NOMA-assisted caching. Second, fixed  NOMA power allocation coefficients have been adopted  in this paper, except for Section  \ref{subsection x}, where    CR inspired power allocation   was used. In general, optimizing the power allocation coefficients is expected to  further improve the performance of NOMA based  caching. Third,     increasing the density of the users or the search area in D2D caching can increase the cache hit probability, but might also cause stronger  interference during the content delivery phase, when the D2D helpers deliver the requested files to their neighbours  simultaneously.  Therefore, for the content delivery phase, it is important to design   low-complexity algorithms for efficient scheduling of the users' requests in order to limit  co-channel interference.  In this context,   coordinated multi-point transmission (CoMP) and cloud radio access networks (C-RAN)  are interesting options for   suppressing   co-channel interference \cite{6665021,7488289}.     Fourth, in this paper, content pushing is carried out without exploiting  the structure of the content files. As shown in  \cite{6763007}, the spectral efficiency of wireless caching can be further improved by applying coded caching, since only parts of the content files need to be cached and one multicast transmission  during the content delivery phase can benefit  multiple users simultaneously. Thus, enhanced caching and delivery schemes combining the benefits of the proposed NOMA  based strategies and coded caching are of interest.  Fifth, in this paper, the cache hit probability has been used    as the performance criterion, whereas   latency is another important metric  \cite{8008769, 7936549, 8002603}.  The proposed caching strategies can potentially decrease  the latency   for content delivery. On the one hand,   the proposed push-then-deliver strategy can effectively reduce the waiting time of the users for being served, since multiple users can be simultaneously served by one content server. On the other hand, with the proposed push-and-deliver strategy, the files cached at the content servers can be updated more frequently, which indirectly  helps in reducing  the latency of content delivery. Hence, a formal analysis  of the impact of the proposed caching strategies on  the latency of content delivery is needed, where  various  effects have to be considered, including the number of retransmissions of the content, the scheduling delay for those users which are served by the BS directly, etc.

\appendices

\section{Proof of Theorem \ref{theorem1}}

Recall that the NOMA cache hit probability is $ \mathrm{P}^{hit}_{m} = \sum^{M_s}_{i=1}\mathrm{P}(f_i) (1-\mathrm{P}_{m,i})$ and the OMA hit probability is $\mathrm{P}^{hit}_{m,OMA} =  \mathrm{P}(f_1) (1-\mathrm{P}^{OMA}_{m,1})$.
Since the file popularity  probabilities are positive and  identical  for the NOMA and OMA cases, proving  $\mathrm{P}^{OMA}_{m,1}=\mathrm{P}_{m,1}$ for all   $\text{CS}_m$, $1\leq m \leq t$, is sufficient to prove the theorem.   

Recall that each content server will carry out SIC, i.e.,   files $j$, $1\leq j<i$, are decoded before   file $i$ is decoded. Therefore,   the outage probability of $\text{CS}_m$ for decoding  file $i$ can be expressed as follows: 
\begin{align}
\mathrm{P}_{m,i}=1-  \mathrm{P}\left( f_j \text{ is decoded}, \forall j\leq i\right).
\end{align}

For notational simplicity, first define $z_{m} \triangleq  \frac{1}{{L\left(||x_m-x_0||\right)}}$, and note that these channel gains are ordered as follows:  $ z_{1}\geq \cdots\geq z_{t}$. Therefore, the outage   probability can be expressed as follows:
\begin{align}
\mathrm{P}_{m,i} = 1- \mathrm{P}\left( z_m>\frac{\epsilon_l}{\rho\xi_l}, \forall l\leq i\right),
\end{align}
where  $\xi_l=\alpha_l^2-\epsilon_l \sum^{{M_s}}_{j=l+1}\alpha_j^2$.  

As discussed previously, showing $\mathrm{P}^{OMA}_{m,1}=\mathrm{P}_{m,1}$ is sufficient to prove the theorem. 
First, we focus on the performance of $\text{CS}_t$. 
According to the definition of the CR NOMA power allocation policy, the outage probability  of $\text{CS}_t$ for decoding the most popular file, $f_1$, is given by
\begin{align}\nonumber
\mathrm{P}_{t,1} &= \mathrm{P}\left( z_t<\frac{\epsilon_1}{\rho\xi_1} \right)
\\  &\overset{(a)}{=} \mathrm{P}\left(  \alpha_1=1 \right)= \mathrm{P}\left( z_t<\frac{\epsilon_1}{\rho } \right)=\mathrm{P}^{OMA}_{t,1},
\end{align}
where $\epsilon_1$ and $\xi_1$ are used since they are valid for file $f_1$.  Step $(a)$ follows from the fact that an outage occurs at $\text{CS}_t$ only if all the power is assigned  to file $f_1$, i.e., $ \alpha_1=1$. 
 Therefore, regarding the   capability of $\text{CS}_t$ to decode $f_1$, adopting   NOMA does not bring any difference, compared to   OMA.  

Second, the outage probability of $\text{CS}_m$, $1\leq m<t$, for decoding $f_1$,  is given by 
\begin{align}\nonumber
\mathrm{P}_{m,1} &= \mathrm{P}\left( z_m<\frac{\epsilon_1}{\rho\xi_1} \right)
\\\nonumber &=\mathrm{P}\left( \alpha_1=1,  z_m<\frac{\epsilon_1}{\rho\xi_1} \right)+\mathrm{P}\left( \alpha_1<1,  z_m<\frac{\epsilon_1}{\rho\xi_1} \right).
\end{align}
Since the channel conditions of $\text{CS}_m$ are better than those of $\text{CS}_t$, the condition that $\text{CS}_t$ can decode $f_1$ correctly, i.e., $\alpha_1<1$, is sufficient to guarantee   successful detection of $f_1$ at $\text{CS}_m$. Therefore, the outage probability can be simplified  as follows: 
\begin{align}\label{zm2}
\mathrm{P}_{m,1}  &=\mathrm{P}\left( \alpha_1=1,  z_m<\frac{\epsilon_1}{\rho\xi_1} \right) .
\end{align}
Note that the use of the CR power allocation policy in \eqref{c1}  complicates the expression for the outage probability, since the power allocation coefficients depend on   the channel conditions of $\text{CS}_t$. 
In order to better understand the outage events, we express  the    event $\{z_m<\frac{\epsilon_1}{\rho\xi_1} \}$   as follows:
\begin{align}\label{zm}
\left\{z_m<\frac{\epsilon_1}{\rho\xi_1}\right\} =& \left\{z_m<\frac{\epsilon_1}{\rho\left(1-P_r-\epsilon_1 P_r\right)}\right\} \\\nonumber =&  \left\{z_m<\frac{\epsilon_1}{\rho\left(1-(1+\epsilon_1 )\max \left\{0,  \frac{ \rho z_t  -\epsilon_1}{\rho (1+\epsilon_1)z_t }\right\}\right)}\right\} 
\\\nonumber =&  \left\{z_m<\frac{\epsilon_1}{\rho\left(1- \max \left\{0,  \frac{ \rho z_t  -\epsilon_1}{\rho z_t }\right\}\right)}\right\} .
\end{align}
By combining \eqref{zm2} and \eqref{zm},  surprisingly probability $\mathrm{P}_{m,1}$ can be  simplified  as follows:
\begin{align}
\mathrm{P}_{m,1}   =\mathrm{P}\left( z_t<\frac{\epsilon_1}{\rho },  z_m<\frac{\epsilon_1}{\rho } \right) , \label{pm1}
\end{align}
since $\max \left\{0,  \frac{ \rho z_t  -\epsilon_1}{\rho (1+\epsilon_1)z_t }\right\}=0$ for the case $ z_t<\frac{\epsilon_1}{\rho\xi_1}$.  On the other hand, it is straightforward to show that the outage probability for OMA is given by 
\begin{align}\nonumber 
\mathrm{P}^{OMA}_{m,1} &= \mathrm{P}\left( z_t>\frac{\epsilon_1}{\rho },  z_m<\frac{\epsilon_1}{\rho } \right)+  \mathrm{P}\left( z_t<\frac{\epsilon_1}{\rho },  z_m<\frac{\epsilon_1}{\rho } \right) \\\nonumber &=\mathrm{P}_{m,1}. 
\end{align}
Therefore, the NOMA outage performance of $\text{CS}_m$, $1\leq m\leq t$, for decoding $f_1$ is the same as that of OMA, but the use of NOMA can ensure that more content is delivered to the content servers, which proves the theorem. 

\section{Proof of Lemma \ref{lemma1}}

Since the content servers follow an HPPP, the pdf for the $m$-th shortest  distance is given by \cite{1512427} 
\begin{align}\label{ordered pdf}
f_{r_m} (x) = \frac{2\lambda_c^m\pi^m x^{2m-1}}{(m-1)!}e^{-\lambda_c \pi x^2} .
\end{align}
The conditional CDF for the $t$-th shortest distance, given $r_m=x$, can be expressed as follows: 
\begin{align}
F_{r_t|r_m}(y) &\triangleq   \mathrm{P}(r_t\leq y|r_m=x)\\\nonumber &=  1 - \mathrm{P}(r_t>y|r_m=x). 
\end{align}
The event, $(r_t>y|r_m=x)$, corresponds to the case  where    the $t$-th nearest content server  is not located inside  the ring between a larger circle with radius $y$ and a smaller one with radius $x$. Or equivalently, the event, $(r_t>y|r_m=x)$,  means that  at most $(t-m-1)$ content servers are inside the ring between the two circles. Therefore, the conditional CDF, $F_{r_t|r_m}(y) $, can be explicitly written  as follows:
\begin{align}
F_{r_t|r_m}(y)  &=1 - \sum^{t-m-1}_{n=0} \mathrm{P}( \# ({\bf int}(\mathcal{B}(x_0, x),\mathcal{B}(x_0, y))=n),
\end{align}
where $\#(\mathcal{A})$ denotes the number of points falling into the area $\mathcal{A}$, $\mathcal{B}(x_0, x)$ denotes a disc with its origin located at $x_0$ and radius $x$, and ${\bf int}(\mathcal{B}(x_0, x),\mathcal{B}(x_0, y))$ denotes the ring between the boundaries of $\mathcal{B}(x_0, x)$ and $\mathcal{B}(x_0, y)$.

By applying the HPPP assumption, the conditional CDF can be found as follows:  
\begin{align}\label{cdfxx}
F_{r_t|r_m}(y)=1 - \sum^{t-m-1}_{n=0}  (\lambda_c \pi)^{n} (y^2-x^2)^{n}\frac{e^{-\lambda_c \pi (y^2-x^2)}}{n!}.
\end{align}
 
In order to find the joint pdf between $r_m$ and $r_t$, the conditional pdf is needed first. However, the  derivative of the CDF $F_{r_t|r_m}(y)$ shown in the above equation has the following complicated form:
\begin{align}\label{diff1}
f_{r_t|r_m}(y)=& \sum^{t-m-1}_{n=1} \frac{ 2y(\lambda_c \pi)^{n}(y^2-x^2)^{n-1}}{n!}e^{-\lambda_c \pi (y^2-x^2)}\\\nonumber &\times 
\left[\lambda_c \pi(y^2-x^2)-n\right] + 2\lambda_c \pi ye^{-\lambda_c \pi (y^2-x^2)}. 
\end{align}
This complicated  form makes the calculation of the outage probability very difficult. Instead, the steps provided in \cite{1512427} can be used   to obtain a much simpler  form, as shown in the following.  First, define $S_n = \frac{(\lambda_c\pi(y^2-x^2))^n}{n!}$, and hence the conditional  CDF obtained in \eqref{cdfxx} can be re-written as follows:
\begin{align}
F_{r_t|r_m}(y)=1 - \sum^{t-m-1}_{n=0} S_n e^{-\lambda_c \pi (y^2-x^2)} .
\end{align}
After taking  the derivative  of the CDF and exploiting  the structure of $S_n$, the conditional pdf can be obtained  as follows: 
\begin{align}\nonumber
f_{r_t|r_m}(y)=& 2y\lambda_c\pi e^{-\lambda_c \pi (y^2-x^2)}\left(\sum^{t-m-1}_{n=0}S_n - \sum^{t-m-1}_{n=1}S_{n-1}\right)
\\  =& 2y(\lambda_c\pi)^{t-m} e^{-\lambda_c \pi (y^2-x^2)}  \frac{ (y^2-x^2)^{t-m-1}}{(t-m-1)!}   ,
\end{align}
which is much simpler than the expression  in \eqref{diff1}. 

By applying Bayes' rule, the joint pdf between $r_m$ and $r_t$ can be obtained as follows:
\begin{align}
f_{r_m,r_t}(x,y) =&f_{r_m|r_t}(x)f_{r_t} (y)\\\nonumber=&  4y(\lambda_c\pi)^{t} e^{-\lambda_c \pi y^2}  \frac{  x^{2m-1}(y^2-x^2)^{t-m-1}}{(t-m-1)!(m-1)!}    . 
\end{align}
Note that, for the special case of $m=t-1$,  the two parameters, $x$ and $y$, are decoupled to yield  the following simplified form for the joint pdf:
 \begin{align}
f_{r_m,r_t}(x,y) &=\frac{4(\lambda_c\pi)^{m+1}yx^{2m-1}}{(m-1)!}e^{-\lambda_c\pi y^2}.
\end{align}
This  completes   the proof of  the lemma.

\section{Proof of Lemma \ref{lemma2}}

Following the steps provided in the proof of Theorem \ref{theorem1},  the outage probability for $\text{CS}_t$ to decode $f_1$ is given by  
\begin{align}\nonumber
\mathrm{P}_{t,1}  = \mathrm{P}\left( z_t<\frac{\epsilon_1}{\rho } \right). \label{alpha1}
\end{align}
After applying the marginal pdf of the $t$-th shortest distance shown in \eqref{ordered pdf}, $\mathrm{P}_{t,1}$ can be calculated as follows:
\begin{align}
\mathrm{P}_{t,1} &=  \frac{2\lambda_c^t\pi^t }{(t-1)!}\int^{\infty}_{\frac{\rho^{\frac{1}{\alpha}}}{\epsilon_1^{\frac{1}{\alpha}}}}y^{2t-1}e^{-\lambda_c \pi y^2} dy\\\nonumber &=e^{-\lambda_c\pi \left(\frac{\rho}{\epsilon_1}\right)^{\frac{2}{\alpha}}}\sum^{t-1}_{k=0} \frac{(\lambda_c\pi)^{k}\left(\frac{\rho}{\epsilon_1}\right)^{\frac{2k}{\alpha}}}{k! }. 
\end{align}

According to \eqref{pm1},   the outage probability for $\text{CS}_m$ to decode $f_1$ is given by 
\begin{align}\nonumber
\mathrm{P}_{m,1}   =\mathrm{P}\left( z_t<\frac{\epsilon_1}{\rho },  z_m<\frac{\epsilon_1}{\rho } \right) . 
\end{align} 
By using the fact that $r_m\leq r_n$ and again applying the marginal distribution of $r_m$, the outage probability can be straightforwardly obtained as follows: 
\begin{align} 
\mathrm{P}_{m,1}  &=\mathrm{P}\left(   z_m<\frac{\epsilon_1}{\rho } \right) 
 =e^{-\lambda_c\pi \left(\frac{\rho}{\epsilon_1}\right)^{\frac{2}{\alpha}}}\sum^{m-1}_{k=0} \frac{(\lambda_c\pi)^{k} \left(\frac{\rho}{\epsilon_1}\right)^{\frac{2k}{\alpha}}}{k! }. \label{pm1x}
\end{align}
Hence, the first part of the lemma is proved. 

The outage probability  for file $i$, $i>1$, is more complicated than the case of $f_1$. The impact of the channel condition of $\text{CS}_t$ on the outage performance of $\text{CS}_m$ can be made explicit  by expressing the  individual  event $\left\{z_m<\frac{\epsilon_i}{\rho\xi_i}\right\}$, $i>1$,  as follows:
\begin{align}\label{pr}
\left\{z_m<\frac{\epsilon_i}{\rho\xi_i}\right\} &= \left\{z_m<\frac{\epsilon_i}{\rho\left(\alpha_i^2-\epsilon_i \sum^{{M_s}}_{j=i+1}\alpha_j^2\right)}\right\} \\\nonumber &= \left\{z_m<\frac{\epsilon_i}{\rho\bar{\xi}_i\max \left\{0,  \frac{ \rho z_t  -\epsilon_1}{\rho (1+\epsilon_1)z_t }\right\}}\right\} ,
\end{align}
where the last step follows from   the fact that $P_r=\max \left\{0,  \frac{ \rho z_t -\epsilon_1}{\rho (1+\epsilon_1)z_t }\right\}$.
Recall that   $\bar{\xi}_i=\left(\beta_i -\epsilon_i \sum^{{M_s}}_{j=i+1}\beta_j  \right)$ is a constant and not a function of the channel conditions of $\text{CS}_t$.  Therefore, the outage probability of  $\text{CS}_t$ for decoding  $f_i$, $i>1$, is given by
\begin{align}
\mathrm{P}_{t,i} =&   \mathrm{P}\left( \alpha_1=1, z_t<\max\left\{\frac{\epsilon_1}{\rho\xi_1},\cdots, \frac{\epsilon_i}{\rho\xi_{i}}\right\} \right) \\\nonumber &+ \mathrm{P}\left( \alpha_1<1, z_t<\max\left\{\frac{\epsilon_1}{\rho\xi_1},\cdots, \frac{\epsilon_i}{\rho\xi_{i}}\right\} \right) .
\end{align}
Note that $\alpha_1=1$ corresponds to the event that all the power is allocated to $f_1$. Therefore, $z_t<\max\left\{\frac{\epsilon_1}{\rho\xi_1},\cdots, \frac{\epsilon_i}{\rho\xi_{i}}\right\} $ is always true if $\alpha_1=1$, and therefore, the outage probability can be simplified  as follows: 
\begin{align}
\mathrm{P}_{t,i} =&   \mathrm{P}\left( \alpha_1=1 \right) \\\nonumber &+ \mathrm{P}\left( \alpha_1<1, z_t<\max\left\{\frac{\epsilon_2}{\rho\xi_2},\cdots, \frac{\epsilon_i}{\rho\xi_{i}}\right\} \right)  .
\end{align}
Note that when $\alpha<1$, the expression for the event $\{z_t<\frac{\epsilon_i}{\rho\xi_i}\}$ in \eqref{pr} can be simplified as follows:
\begin{align}\label{pr1}
\left\{z_t<\frac{\epsilon_i}{\rho\xi_i}\right\}   &= \left\{z_t<\frac{\epsilon_i}{\rho\bar{\xi}_i   \frac{ \rho z_t  -\epsilon_1}{\rho (1+\epsilon_1)z_t } }\right\} .
\end{align}
Therefore, the outage probability can be rewritten as follows:
\begin{align}
\mathrm{P}_{t,i} =&   \mathrm{P}\left( \alpha_1=1 \right) \\\nonumber &+ \mathrm{P}\left( \alpha_1<1, z_t<\max\left\{\frac{\epsilon_j}{\rho\bar{\xi}_j   \frac{ \rho z_t  -\epsilon_1}{\rho (1+\epsilon_1)z_t } },2\leq j\leq i \right\} \right) \\\nonumber
&=   \mathrm{P}\left( z_t<\frac{\epsilon_1}{\rho } \right) + \mathrm{P}\left( z_t>\frac{\epsilon_1}{\rho }, z_t< \frac{\epsilon_1}{\rho} +\frac{ (1+\epsilon_1)}{\rho \phi_i} \right) .
\end{align}
By applying the marginal pdf for the $t$-th shortest distance, the outage probability for $\text{CS}_t$ to decode $f_i$ can be obtained as follows:
\begin{align}
\mathrm{P}_{t,i} & =e^{-\lambda_c\pi \left( \frac{\epsilon_1}{\rho} +\frac{(1+\epsilon_1)}{\rho \phi_i}\right)^{-\frac{2}{\alpha}}}\sum^{t-1}_{k=0} \frac{(\lambda_c\pi)^{k}\left( \frac{\epsilon_1}{\rho} +\frac{(1+\epsilon_1)}{\rho \phi_i}\right)^{-\frac{2k}{\alpha}}}{k! }.
\end{align}
Hence, the second part of the lemma is proved. 

The outage probability for  $\text{CS}_m$ to decode $f_i$, $i>1$, is the most difficult to obtain among the probabilities  shown in the lemma. This probability can be first expressed as follows: 
\begin{align}
\mathrm{P}_{m,i} =&   \mathrm{P}\left( \alpha_1=1, z_m<\max\left\{\frac{\epsilon_1}{\rho\xi_1},\cdots, \frac{\epsilon_i}{\rho\xi_{i}}\right\} \right) \\\nonumber &+ \mathrm{P}\left( \alpha_1<1, z_m<\max\left\{\frac{\epsilon_1}{\rho\xi_1},\cdots, \frac{\epsilon_1}{\rho\xi_{i}}\right\} \right) .
\end{align}
Note that $\alpha_1=1$ results in the situation that no power is allocated to $f_j$, $j>1$, which means that the event $z_m<\max\left\{\frac{\epsilon_1}{\rho\xi_1},\cdots, \frac{\epsilon_i}{\rho\xi_{i}}\right\} $ always happens, if $\alpha_1=1$. In addition, by using the fact that $r_m\leq r_t$, the outage probability can be simplified as follows:
\begin{align}\label{pQ}
\mathrm{P}_{m,i} =&   \mathrm{P}\left(z_t<\frac{\epsilon_1}{\rho } \right) \\\nonumber &+ \underset{Q_1}{\underbrace{\mathrm{P}\left( z_t>\frac{\epsilon_1}{\rho }, z_m<\max\left\{\frac{\epsilon_2}{\rho\xi_2},\cdots, \frac{\epsilon_i}{\rho\xi_{i}}\right\} \right) }}.
\end{align}
Note that $z_t>\frac{\epsilon_1}{\rho }$   guarantees $z_m> \frac{\epsilon_1}{\rho\xi_1}$, as   $z_t\leq z_m$ and   $z_t> \frac{\epsilon_1}{\rho\xi_1}$ is equivalent to $z_t>\frac{\epsilon_1}{\rho }$. However, $z_t>\frac{\epsilon_1}{\rho }$ does not  guarantee $z_m>\frac{\epsilon_j}{\rho\xi_j}$, $j>1$. Recall that conditioned on  $z_t>\frac{\epsilon_1}{\rho }$, the term $\frac{\epsilon_j}{\rho\xi_j}$, $j>1$, can be simplified  as follows:
\begin{align}
\frac{\epsilon_i}{\rho\xi_i}  &= \frac{\epsilon_i}{ \bar{\xi}_i \frac{ \rho z_t  -\epsilon_1}{  (1+\epsilon_1)z_t }}.
\end{align}

Therefore, the term $Q_1$ can be calculated as follows:
\begin{align}
Q_1&=\mathrm{P}\left( z_t>\frac{\epsilon_1}{\rho }, z_m<\max\left\{\frac{\epsilon_2}{\rho\xi_2},\cdots, \frac{\epsilon_i}{\rho\xi_{i}}\right\} \right)\\\nonumber&=
\mathrm{P}\left( z_t>\frac{\epsilon_1}{\rho }, z_m<\frac{ (1+\epsilon_1)}{ \phi_i \left( \rho   -\frac{\epsilon_1}{  z_t }\right)}\right).
\end{align}
After applying the path loss model, $z_t$ ($z_m$) can be replaced by the distance between the BS and $\text{CS}_t$ ($\text{CS}_m$), and  the outage probability can be expressed as follows: 
\begin{align}\label{Q1 last}
Q_1&= 
\mathrm{P}\left( y<\left(\frac{\epsilon_1}{\rho }\right)^{-\frac{1}{\alpha}}, x>\left(\frac{(1+\epsilon_1)}{ \phi_i \left( \rho   -\epsilon_1 y^{\alpha}\right)}\right)^{-\frac{1}{\alpha}}\right),
\end{align}
where $x$ denotes the   distance between the BS and $\text{CS}_t$ and $y$ denotes the distance between the BS and $\text{CS}_m$. 
However,  there is an extra constraint on $y$ as follows:
\begin{align}
\left(\frac{\epsilon_1}{\rho }\right)^{-\frac{1}{\alpha}}>\left(\frac{(1+\epsilon_1)}{ \phi_i \left( \rho   -\epsilon_1 y^{\alpha}\right)}\right)^{-\frac{1}{\alpha}} ,
\end{align}
which leads to the following constraint on $y$:
\begin{align}\label{constraint1}
y^\alpha>  \frac{\rho}{\epsilon_1}\left[1-\frac{1+\epsilon_1}{\epsilon_1\phi_i}\right] .
\end{align}
To better understand     this constraint, the term $\frac{1+\epsilon_1}{\epsilon_1\phi_i}$ is rewritten as follows:
\begin{align}\label{termx1}
\frac{1+\epsilon_1}{\epsilon_1\phi_i} =&  \frac{1+\epsilon_1}{\epsilon_1\min\left\{\frac{\bar{\xi}_2}{\epsilon_2}, \cdots, \frac{\bar{\xi}_{M_s}}{\epsilon_{M_s}}\right\}} \geq  \frac{1+\epsilon_1}{\epsilon_1 \frac{\bar{\xi}_{M_s}}{\epsilon_{M_s}} } = \frac{\epsilon_{M_s} 2^{R_1} }{\epsilon_1 \bar{\xi}_{M_s} } ,
\end{align}
 where  $\bar{\xi}_{M_s}\leq 1$ and $2^{R_1}\geq 1$ hold. The only uncertainty for the comparison between the term $\frac{1+\epsilon_1}{\epsilon_1\phi_i}$ and $1$ is caused by the relationship between $\epsilon_1$ and $\epsilon_{M_s}$. In the lemma, it is assumed that    $\epsilon_1\leq \epsilon_{M_s}$.  As a result,  the constraint in \eqref{constraint1} is always satisfied 
  since $\frac{1+\epsilon_1}{\epsilon_1\phi_i}\geq 1$. However, the probability in \eqref{Q1 last} also implies the following constraint:
\begin{align}
y>  \left(\frac{(1+\epsilon_1)}{ \phi_i \left( \rho   -\epsilon_1 y^{\alpha}\right)}\right)^{-\frac{1}{\alpha}}.
\end{align}
 This leads to the following constraint on $y$:
\begin{align}
y > \left(\frac{\rho \phi_i}{1+\epsilon_1+\epsilon_1\phi_i}\right)^{\frac{1}{\alpha}}\triangleq \tau_1.
\end{align}

After understanding the ranges of $x$ and $y$,  we can now apply the joint pdf to calculate the outage probability,  which yields the following:
\begin{align}
Q_1&= \int_{\tau_1}^{\tau_2}\int^{y}_{\left(\frac{(1+\epsilon_1)}{ \phi_i \left( \rho   -\epsilon_1 y^{\alpha}\right)}\right)^{-\frac{1}{\alpha}}} f_{r_m,r_t}(x,y)dx dy,
\end{align}
where $\tau_2$ is defined in the lemma. 
To facilitate the calculation of this integral, the joint pdf is rewritten as follows:
\begin{align}
f_{r_m,r_t}(x,y)=& \frac{ 4(\lambda_c\pi)^{t}}{(t-m-1)!(m-1)!}  e^{-\lambda_c \pi y^2}  \sum^{t-m-1}_{p=0} (-1)^p \\\nonumber &\times  {t-m-1\choose p} y^{2(t-m-1)-2p+1}      x^{2m+2p-1}.
\end{align}   
 Now, we can apply the joint pdf which yields the following:
\begin{align}  
Q_1=&  \frac{ 4(\lambda_c\pi)^{t}}{(t-m-1)!(m-1)!}    \sum^{t-m-1}_{p=0}(-1)^p{t-m-1 \choose p} \\\nonumber &\times \int_{\tau_1}^{\tau_2}f_m(y) dy  ,
\end{align}
where $f_m(\cdot)$ is defined in the lemma. 
One can apply  Chebyshev-Gauss quadrature to obtain the following expression for $Q_1$: 
\begin{align}\label{Q1 final} 
Q_1&\approx  \frac{ 4(\lambda_c\pi)^{t}}{(t-m-1)!(m-1)!}    \sum^{t-m-1}_{p=0}(-1)^p{t-m-1 \choose p} \\\nonumber &\times  \sum^{N}_{l=1}\frac{\pi\left(\tau_2-\tau_1\right)}{2 N} f_m\left(\frac{\tau_2-\tau_1}{2}w_l+\frac{\tau_2+\tau_1}{2}\right)\sqrt{1-w_l^2}.
\end{align}
Substituting  \eqref{Q1 final} and  \eqref{pm1x} into \eqref{pQ}, the third part of the lemma is proved.  

\section{Proof of Lemma \ref{lemma3}}

  Because the two users associated with the same content server are located in different regions inside the disc with radius $\mathcal{R}_c$, the density functions for their channel gains are different, and therefore, the two users' outage probabilities will be calculated separately in the following subsections.  
  
  \subsubsection{The outage performance at $\text{U}_{m,2}$} First define the composite channel gain as $z_{m,k} \triangleq \frac{|h_{m,mk}|^2}{L\left(||y_{m,k}||\right)}$, for $k\in\{1, 2\}$.  Recall that, for a user which is uniformly distributed in a disc with radius $r$, the CDF of its composite channel gain which includes  the effects of  small scale Rayleigh fading and path loss can be expressed as follows \cite{Nomading}:
\begin{eqnarray}\label{cdf11}
F_{r}(z) \approx   \sum^{N}_{n=1}\bar{w}_n  \left(1-e^{-c_{n,r}z}\right),
\end{eqnarray}
and the corresponding pdf of the channel gain is $f_{r}(z) \approx  \sum^{N}_{n=1}\bar{w}_n c_{n,r} e^{-c_{n,r}z}$.   Recall that $\text{U}_{m,2}$ is uniformly distributed in a disc with radius $\mathcal{R}_s$, and therefore, the CDF and pdf of the channel gain of $\text{U}_{m,2}$ are simply given by  $F_{\mathcal{R}_s}(z)$ and $f_{\mathcal{R}_s}(z)$ by replacing $r$ with $\mathcal{R}_s$.  The reason for using   the approximated form in \eqref{cdf11} is   that both the approximated CDF and pdf  are in the form of exponential functions. In the following, we will show that these exponential functions will signficiantly simplify the application of  the probability generating functional (PGFL). 

With the definition of  $z_{m,k} \triangleq \frac{|h_{m,mk}|^2}{L\left(||y_{m,k}||\right)}$, the SINR of $\text{U}_{m,2}$ for decoding  the first message, $f_{m,1}$, is given by
\begin{align}
 \text{SINR}_{m,2}^1 = \frac{\alpha^2_1z_{m,2}}{ \alpha^2_2z_{m,2}     +\text{I}^{m,2}_{inter}
 +\frac{1}{\rho}}.
\end{align}
Similarly,  the SINR of $\text{U}_{m,2}$ for decoding  its own message, $f_{m,2}$, can be rewritten as follows: 
\begin{align}
 \text{SINR}_{m,2}^2 = \frac{\alpha^2_2z_{m,2}}{ \text{I}^{m,2}_{inter}
 +\frac{1}{\rho}}.
\end{align}
Therefore,  the outage probability of $\text{U}_{m,2}$ for decoding its own  message can be expressed as follows:
\begin{align}
\mathrm{P}^o_{m,2}  =& 1 - \mathrm{P}\left(\log(1+\text{SINR}^l_{m,2})>R_l, l\in\{1,2\} \right)\\ \nonumber
 =& \mathcal{E}_{\text{I}^{m,2}_{inter}}\left\{ \mathrm{P}\left( z_{m,2}<\max\left\{\frac{\epsilon_1 \text{I}^{m,2}_{inter}+\frac{\epsilon_1}{\rho}}{ \alpha^2_1-\epsilon_1\alpha^2_2 },  \frac{\epsilon_2 \text{I}^{m,2}_{inter}+\frac{\epsilon_2}{\rho}}{  \alpha^2_2 }\right\}\right) \right\},
\end{align}
where $ \mathcal{E}_{x}\{\cdot\}$ denotes the expectation operation  with respect to $x$.
In order to facilitate the application of the PGFL, the outage probability  is first rewritten as follows:
\begin{align}
\mathrm{P}^o_{m,2}   
 =& \mathcal{E}_{\text{I}^{m,2}_{inter}}\left\{ \mathrm{P}\left( z_{m,2}<\max\left\{\frac{ \text{I}^{m,2}_{inter}+\frac{1}{\rho}}{ \frac{\alpha^2_1-\epsilon_1\alpha^2_2}{\epsilon_1} },  \frac{ \text{I}^{m,2}_{inter}+\frac{1}{\rho}}{ \frac{ \alpha^2_2 }{\epsilon_2}}\right\}\right) \right\}\\\nonumber
  =& \mathcal{E}_{\text{I}^{m,2}_{inter}}\left\{ \mathrm{P}\left( z_{m,2}<\frac{ \text{I}^{m,2}_{inter}+\frac{1}{\rho}}{\min\left\{ \frac{\alpha^2_1-\epsilon_1\alpha^2_2}{\epsilon_1} ,  \frac{ \alpha^2_2 }{\epsilon_2}\right\}}\right) \right\}. 
\end{align}

After using the approximated expression for the pdf of $z_{m,2}$, the outage probability can be approximated as follows: 
\begin{align}
\mathrm{P}^o_{m,2}  \approx&  \mathcal{E}_{\text{I}^{m,2}_{inter} }\left\{ \sum^{N}_{n=1}\bar{w}_n  \left(1-e^{-c_{n,\mathcal{R}_s}\frac{  \text{I}^{m,2}_{inter}+\frac{1}{\rho}}{\min\left\{ \frac{\alpha^2_1-\epsilon_1\alpha^2_2}{\epsilon_1} ,  \frac{ \alpha^2_2 }{\epsilon_2}\right\}}}\right) \right\}
 \\\nonumber
 \approx& 1- \sum^{N}_{n=1}\bar{w}_ne^{-\frac{c_{n,\mathcal{R}_s}\frac{1}{\rho}}{\min\left\{ \frac{\alpha^2_1-\epsilon_1\alpha^2_2}{\epsilon_1} ,  \frac{ \alpha^2_2 }{\epsilon_2}\right\}}}   \mathcal{E}_{\text{I}^{m,2}_{inter} }\left\{ e^{-\frac{c_{n,\mathcal{R}_s} \text{I}^{m,2}_{inter}}{\min\left\{ \frac{\alpha^2_1-\epsilon_1\alpha^2_2}{\epsilon_1} ,  \frac{ \alpha^2_2 }{\epsilon_2}\right\}}}  \right\} .
 \end{align}

Denote the Laplace transform of   $\text{I}^{m,2}_{inter}$ by $\mathcal{L}_{\text{I}^{m,2}_{inter}}(s)$.  Then, the outage probability can be  rewritten as follows:
\begin{align}\label{eqxc3}
\mathrm{P}^o_{m,2}  \approx&     1- \sum^{N}_{n=1}\bar{w}_ne^{-\frac{c_{n,\mathcal{R}_s}\frac{1}{\rho}}{\min\left\{ \frac{\alpha^2_1-\epsilon_1\alpha^2_2}{\epsilon_1} ,  \frac{ \alpha^2_2 }{\epsilon_2}\right\}}}   \mathcal{L}_{\text{I}^{m,2}_{inter}}\left(\frac{c_{n,\mathcal{R}_s}  }{\min\left\{ \frac{\alpha^2_1-\epsilon_1\alpha^2_2}{\epsilon_1} ,  \frac{ \alpha^2_2 }{\epsilon_2}\right\}}\right)  .
 \end{align} 
 Therefore, the outage probability can be calculated if  the Laplace transform of   $\text{I}^{m,2}_{inter}$ is known.  
 Particularly, the Laplace transform of   $\text{I}^{m,2}_{inter}$,  $\mathcal{L}_{\text{I}^{m,2}_{inter}}(s)$, can be first expressed   as follows:
\begin{align}\nonumber 
\mathcal{L}_{\text{I}^{m,2}_{inter}}(s) &= \mathcal{E} \left\{ \prod_{x_j\in \Phi_c\backslash x_m} {\rm exp}\left( -s \frac{|h_{j,m2}|^2}{ {L\left(||y_{m,2}+x_m-x_j||\right)}}\right)\right\}.
\end{align}
By using the assumption that $h_{j,m2}$ is Rayleigh distributed, the small scale fading gain can be averaged out in the expression, and the Laplace transform can be expressed as follows:
\begin{align} 
\mathcal{L}_{\text{I}^{m,2}_{inter}}(s) &= \mathcal{E} \left\{ \prod_{x_j\in \Phi_c\backslash x_m} \frac{1}{ \frac{s}{ {L\left(||y_{m,2}+x_m-x_j||\right)}}+1}\right\}.
\end{align} 

By applying the Campell theorem and the PFGL \cite{Haenggi, 5560889, 7996589}, the Laplace transform can be simplified as follows:
\begin{align} 
\mathcal{L}_{\text{I}^{m,2}_{inter}}(s) =& {\rm exp}\left(-\lambda_{c} \int_{\mathbb{R}^2}\left(1 -\mathcal{E}_{y_{m,2}}\left\{\frac{1}{ \frac{s}{ {L\left(||y_{m,2}+x_m-x||\right)}}+1}\right\} \right) dx \right),
\end{align}
which contains a 2-D integral with respect to a HPPP point $x$. 
Denote the pdf of $y_{m,2}$, $y_{m,2}\in \mathcal{B}(x_m,\mathcal{R}_s)$, by $f_{y_{m,2}} (y)$, where we recall that $\mathcal{B}(x_m,\mathcal{R}_s)$ denotes the   disc with radius $\mathcal{R}_s$ and its  origin   located at $x_m$. Therefore, the Laplace transform can be expressed as follows:  
\begin{align} 
\mathcal{L}_{\text{I}^{m,2}_{inter}}(s) =& {\rm exp}\left(-\lambda_{c} \int_{\mathcal{B}(x_m,\mathcal{R}_s)}f_{y_{m,2}} (y) \int_{\mathbb{R}^2} \left(1   - \frac{1}{ \frac{s}{ {L\left(||y +x_m-x||\right)}}+1} \right) dx dy\right).
\end{align}
Following  similar  steps as  in \cite{Haenggi, 5560889, 7110502,7996589}, the substitution of $y+x_m-x\rightarrow x'$ can be used to simplify the expression of the Laplace transform as follows: 
\begin{align}
\mathcal{L}_{\text{I}^{m,2}_{inter}}(s) =& {\rm exp}\left(-\lambda_{c} \int_{\mathcal{B}(x_m,\mathcal{R}_s)}f_{y_{m,2}} (y) \int_{\mathbb{R}^2} \left(1  - \frac{1}{ \frac{s}{ {L\left(||x'||\right)}}+1}  \right) dx'dy \right)\\ \nonumber  =& {\rm exp}\left(-\lambda_{c} \int_{\mathcal{B}(x_m,\mathcal{R}_s)}f_{y_{m,2}} (y) 2\pi \int_0^\infty  \left(1  - \frac{1}{ \frac{s}{ {L\left(r\right)}}+1} \right) rdrdy \right). 
\end{align}
After applying the  Beta function \cite{GRADSHTEYN}, the Laplace transform can be obtained as follows:
 \begin{align}\nonumber
\mathcal{L}_{\text{I}^{m,2}_{inter}}(s)   =& {\rm exp}\left(-\lambda_{c} \int_{\mathcal{B}(x_m,\mathcal{R}_s)}f_{y_{m,2}} (y) 2\pi \frac{s^{\frac{2}{\alpha}}}{\alpha}\text{B}\left(\frac{2}{\alpha}, \frac{\alpha-2}{\alpha}\right)dy \right)
\\  =& \label{eqxc2}
{\rm exp}\left(-   2\pi\lambda_{c} \frac{s^{\frac{2}{\alpha}}}{\alpha}\text{B}\left(\frac{2}{\alpha}, \frac{\alpha-2}{\alpha}\right) \right),
\end{align}
where the last equality follows from the fact that the integral with respect to $y$ is not a function of $x$.  Substituting  \eqref{eqxc2} into \eqref{eqxc3}, the first part of the lemma is proved.

\subsubsection{The outage performance at $\text{U}_{m,1}$} Recall that  $\text{U}_{m,1}$ is located inside a ring with $ \mathcal{R}_s$ as the inner radius and $\mathcal{R}_c$ as the outer radius. Therefore,  the CDF of this user's channel gain needs to  be calculated differently compared to that of  $\text{U}_{m,2}$. First, by using the assumptions that the user is uniformly distributed inside the ring and the fading gain is Rayleigh distributed,   the CDF of $z_{m,1}$ can be expressed as follows \cite{Wangpoor11}:
\begin{align}
F_{z_{m,1}}(z) =&\frac{2}{\mathcal{R}_c^2-\mathcal{R}_s^2}\int^{\mathcal{R}_c}_{\mathcal{R}_s} \left(1- e^{-r^\alpha z}\right)rdr\\\nonumber = &  \frac{1}{\mathcal{R}_c^2-\mathcal{R}_s^2}\left[\mathcal{R}_c^2\frac{2}{\mathcal{R}_c^2}\int^{\mathcal{R}_c}_{0} \left(1- e^{-r^\alpha z}\right)rdr\right. \\\nonumber &\left.-\mathcal{R}_s^2\frac{2}{\mathcal{R}_s^2}\int^{\mathcal{R}_s}_0 \left(1- e^{-r^\alpha z}\right)rdr\right].
\end{align}
Comparing this with   \cite[Eq. (3)]{Nomading}, one can find that the approximated form shown in \eqref{cdf11} can be applied to each term in the above expression, and hence the CDF can be approximated as follows: 
\begin{align}
F_{z_{m,1}}(z) =  &  \frac{1}{\mathcal{R}_c^2-\mathcal{R}_s^2}\left[\mathcal{R}_c^2F_{\mathcal{R}_c}(z) -\mathcal{R}_s^2F_{\mathcal{R}_s}(z)\right].
\end{align}
Following  similar  steps as in the previous subsection, the outage probability of $\text{U}_{m,1}$ for decoding $f_{m,1}$ can be obtained as follows\begin{align}
\mathrm{P}^o_{m,1}   
 =& \mathcal{E}_{\text{I}^{m,1}_{inter}}\left\{ \mathrm{P}\left( z_{m,1}< \frac{\epsilon_1 \text{I}^{m,1}_{inter}+\frac{\epsilon_1}{\rho}}{\alpha^2_1-\epsilon_1\alpha^2_2}\right) \right\}. 
\end{align} 
After using the approximated expression for the pdf of $z_{m,1}$, the outage probability can be approximated   as follows:
\begin{align}\nonumber
\mathrm{P}^o_{m,1}  \approx& \frac{\mathcal{R}_c^2}{\mathcal{R}_c^2-\mathcal{R}_s^2}  \mathcal{E}_{\text{I}^{m,1}_{inter} }\left\{ \sum^{N}_{n=1}\bar{w}_n  \left(1-e^{-c_{n,\mathcal{R}_c}\frac{\epsilon_1 \text{I}^{m,1}_{inter}+\frac{\epsilon_1}{\rho}}{\alpha^2_1-\epsilon_1\alpha^2_2}}\right) \right\}\\\nonumber  &-
 \frac{\mathcal{R}_s^2}{\mathcal{R}_c^2-\mathcal{R}_s^2}  \mathcal{E}_{\text{I}^{m,1}_{inter} }\left\{ \sum^{N}_{n=1}\bar{w}_n  \left(1-e^{-c_{n,\mathcal{R}_s}\frac{\epsilon_1 \text{I}^{m,1}_{inter}+\frac{\epsilon_1}{\rho}}{\alpha^2_1-\epsilon_1\alpha^2_2}}\right) \right\}
 \\  \label{eqxc1}
 \approx& 1+\frac{\mathcal{R}_s^2}{\mathcal{R}_c^2-\mathcal{R}_s^2} \sum^{N}_{n=1}\bar{w}_ne^{-\frac{c_{n,\mathcal{R}_s}\frac{\epsilon_1}{\rho}}{\alpha^2_1-\epsilon_1\alpha^2_2}}   \mathcal{E}_{\text{I}^{m,1}_{inter} }\left\{ e^{-\frac{c_{n,\mathcal{R}_s}\epsilon_1 \text{I}^{m,1}_{inter}}{\alpha^2_1-\epsilon_1\alpha^2_2}}  \right\} \\\nonumber &- \frac{\mathcal{R}_c^2}{\mathcal{R}_c^2-\mathcal{R}_s^2} \sum^{N}_{n=1}\bar{w}_ne^{-\frac{c_{n,\mathcal{R}_c}\frac{\epsilon_1}{\rho}}{\alpha^2_1-\epsilon_1\alpha^2_2}}   \mathcal{E}_{\text{I}^{m,1}_{inter} }\left\{ e^{-\frac{c_{n,\mathcal{R}_c}\epsilon_1 \text{I}^{m,1}_{inter}}{\alpha^2_1-\epsilon_1\alpha^2_2}}  \right\} .
 \end{align}
It is straightforward to show that the Laplace transform of $\text{I}^{m,1}_{inter}$ is the same as that of $\text{I}^{m,2}_{inter}$. Therefore, substituting  \eqref{eqxc2} with \eqref{eqxc1}, the second part of the lemma is proved. 

   \bibliographystyle{IEEEtran}
\bibliography{IEEEfull,trasfer}

  \end{document}